\documentclass[11pt]{article}

\usepackage{blindtext}
\usepackage{graphicx} 
\usepackage{amsmath}
\usepackage{natbib}
\usepackage{xcolor}
\usepackage{bbm, dsfont}
\usepackage{prodint}
\usepackage{enumitem}
\usepackage{algorithm}
\usepackage{subcaption}
\usepackage[noend]{algpseudocode}

\usepackage[utf8]{inputenc}
\usepackage[T1]{fontenc}
\usepackage{lmodern}          
\usepackage{amsmath,amssymb}
\usepackage{graphicx}
\usepackage{natbib}
\usepackage[margin=1in]{geometry}

\usepackage[colorlinks=false,allbordercolors={1 1 1}]{hyperref}

\usepackage{amsthm}
\newtheorem{theorem}{Theorem}
\newtheorem{lemma}[theorem]{Lemma}

\theoremstyle{remark}

\newtheorem{proposition}[theorem]{Proposition}

\newcommand{\acks}[1]{%
  \section*{Acknowledgments and Disclosure of Funding}%
  #1\par
}

\newtheorem{condition}{Condition}

\usepackage{lastpage}

\begin{document}

\title{Training-Set Conditionally Valid Prediction Sets with Right-Censored Data}



\author{%
  Wenwen Si\\
  University of Pennsylvania\\
  \texttt{wenwens@seas.upenn.edu}
  \and
  Hongxiang Qiu\\
  Michigan State University\\
  \texttt{qiuhongx@msu.edu}
}

\date{\today}

\maketitle

\begin{abstract}
Uncertainty quantification of prediction models through prediction sets is increasingly popular and successful, but most existing methods rely on directly observing the outcome and do not appropriately handle censored outcomes, such as time-to-event outcomes.
    \citet{candes2023conformalized} and \citet{gui2024conformalized} have introduced distribution-free conformal prediction methods that construct predictive intervals for right-censored outcomes with 
    marginal coverage guarantees. However, these methods typically assume a restrictive Type I censoring framework, in which censoring times are 
    all observed.

    In this paper, we leverage a semiparametric one-step estimation framework and propose a novel approach for constructing predictive lower bounds on survival times with training-set conditional validity under right-censoring, where censoring times may be unobserved when the survival time is observed.
With slight modification, our method can also provide predictive lower bounds with marginal guarantees.
Through extensive simulations and a real-world application dataset tracking users' active times on a mobile application, we demonstrate the effectiveness and practicality of our approach. Compared to existing methods, our technique shows superior efficiency and robustness to model misspecifications, representing a significant advancement in the toolbox for reliable machine learning for time-to-event outcomes.
\end{abstract}

\section{Introduction}

Time-to-event outcomes arise in a wide range of applications, including clinical trials \citep{fleming2000survival, singh2011survival}, ecology \citep{muenchow1986ecological}, and economics \citep{powell1986censored, hong2003inference}. 
Time-to-event outcomes are often subject to right-censoring; that is, the follow-up may terminate at a censoring time before the event occurs, in which case the time-to-event is only known to exceed the censoring time but is not directly observed.
Such incomplete observation poses a unique challenge compared to fully observed outcomes.
A variety of survival analysis methods have been developed to analyze right-censored data, including the Kaplan-Meier (KM) estimator \citep{kaplan1958nonparametric} and Cox proportional hazards model \citep{cox1972regression}, among others.

Despite the success of these methods, uncertainty quantification for right-censored data remains challenging.  
For fully observed (i.e., non-censored) outcomes, a widely recognized method for distribution-free uncertainty quantification is conformal prediction (CP) \citep{saunders1999transduction, vovk1999machine, vovk2005algorithmic, chernozhukov2018double, dunn2018distribution, lei2014distribution, lei2013distribution, lei2018distribution}. CP guarantees a high probability of covering a new observation, where this probability is assessed marginally over both the calibration data and the new observation. 
With slight modification, split CP (also known as inductive CP) can also provide prediction sets with training-set conditional (i.e., PAC) guarantees \citep{vovk2012conditional}.
However, extending CP to censored data is nontrivial.

Some recent work provides solutions with conformal-based inference. Specifically, \citet{candes2023conformalized} and \citet{gui2024conformalized} consider Type I censoring, where all censoring times are observed. They treat censoring as a distribution shift and propose methods to compute lower prediction bounds (LPBs) with marginal guarantees.
\citet{sesia2024doubly} extends their setting to the right-censored data by imputing unobserved censoring times using a suitable model for the censoring time. All these methods share a desirable property called \emph{double robustness} \citep{bang2005doubly}, meaning they are robust against inconsistent estimation of one of the two nuisance functions as long as the other is estimated accurately. \citep{holmes2024two} use conformal prediction to generate two-sided predictive bounds for individuals deemed sufficiently similar to the non-censored population. Most relevant to our work, \citet{gui2024conformalized} and concurrent work \citet{farina2025doubly} achieve the probably asymptotically approximately correct (PAAC) \citep{qiu2023prediction} in Eq.~\eqref{eq:paac}, an asymptotic training-set conditional validity.

Beyond CP-based methods that heavily rely on exchangeability (possibly after weighting),
prediction sets can be constructed based on estimators of coverage.
In particular, for more complicated settings such as distribution shift and censoring, semiparametric efficiency theory \citep{werner2004lectures, pfanzagl1985contributions, pfanzagl1990estimation, van2000asymptotic} provides a flexible framework to estimate coverage. 
Specifically, one-step estimators based on efficiency theory can asymptotically efficiently estimate the coverage of prediction sets and thus help construct prediction sets with asymptotic marginal coverage or PAC-style guarantees \citep{Yang2022,qiu2023prediction,qiu2024efficient}. In the context of survival analysis, several doubly robust estimators have been developed to estimate treatment-specific survival curves in continuous time \citep{hubbard2000nonparametric,bai2013doubly,westling2023inference}. 

In this paper, we propose a novel method to construct lower prediction bounds for right-censored data based on an asymptotically efficient one-step corrected estimator of coverage error, achieving training-set conditional validity. 
In particular, our method is the first to achieve the asymptotic probably approximately correct (APAC) guarantee \citep{qiu2023prediction} in Eq.~\eqref{eq: apac}, another form of asymptotic training-set conditional validity, which may be more desirable than PAAC in Eq.~\ref{eq:paac} for safety purposes. 
Building on \citet{hubbard2000nonparametric,bai2013doubly,westling2023inference}, we derive a closed-form expression for the efficient influence function of the covariate-dependent LPB, enabling efficient nonparametric inference. 
Our method can also be easily modified to achieve marginal validity.
Notably, our estimator remains consistent even if either the censoring or survival models is inconsistently estimated, thereby ensuring double robustness of marginal validity.
In simulations, our method achieves superior performance compared to existing methods in \citet{candes2023conformalized, gui2024conformalized}, regardless of whether all censoring times are observed.

The remainder of this paper is structured as follows. Section~\ref{sec:2} introduces the problem setup. Section~\ref{sec:3} presents our proposed method, TCsurv (Training-set Conditional survival), including key identification results and the derivation of the efficient influence function for population-level miscoverage. Section~\ref{sec:4} establishes the theoretical properties of TCsurv. Theorem~\ref{thm3} shows that TCsurv constructs asymptotically PAC prediction sets for asymptotically linear estimators by controlling nominal vs. realized coverage discrepancies in Wald confidence intervals. Section~\ref{sec:5} assesses empirical performance through simulations, and Section~\ref{sec:6} illustrates real-world applicability.

\section{Preliminaries}\label{sec:2}

In this section, we first describe right-censored data in survival analysis. We then provide a formal definition of the lower prediction bound with asymptotically probably approximately correct validity.

\subsection{Right-censored data}
Let $(W, C, T)$ denote the prototypical full data point, where $W \in \mathcal{W} \subseteq \mathbb{R}^p$ is the $p$-dimensional covariate vector, $T \in \mathbb{R}_{\geq 0}$ is the survival time, and $C \in \mathbb{R}_{\geq 0}$ is the censoring time.
Right censoring occurs when the event time $T$ is not fully observed because it exceeds $C$, so that only a follow-up time $Y := T \wedge C$ is recorded.
Here, $\wedge$ denotes the minimum of two real numbers.
For example, in a clinical trial with a fixed end date, $W$ may include covariates such as age and gender, $C$ represents the time from enrollment to the trial’s end, and a patient is either censored ($T > C$) or experiences the event ($T \leq C$). 
Let $P_0$ denote the true, unknown distribution of $(W,C,T)$.

There are two common right-censoring settings, distinguished by whether the censoring times $C$ are observed. In Type I censoring, as considered by \citet{candes2023conformalized, gui2024conformalized}, all censoring times $C$ are fully observed. In this case, for each individual $i$, we observe the tuple $(W_i, C_i, Y_i)$, where $Y_i = T_i \wedge C_i$. This setting commonly arises in clinical trials, in which participants remain in the study until a predetermined end date.

However, in many other practical scenarios, censoring may occur due to spontaneous dropout from the study, so that censoring times $C$ are not always observed. 
In this case, for each individual $i$, we observe the tuple $O_i = (W_i, \Delta_i, Y_i)$, where $\Delta_i = \mathbbm{1}[T_i \leq C_i]$
indicates whether the event occurred before censoring. 
Note that $P_0$ uniquely determines the distribution of the observed data point $O$ in this setting.
We consider this setting throughout the paper and assume that $N$ independently and identically distributed (i.i.d.) copies of $O_i$ are observed.

\subsection{Lower prediction bound}
We first split the data $\mathcal{D}$ into a training set $\mathcal{D}_{\mathrm{train}}$ and a calibration set $\mathcal{D}_{\mathrm{cal}}$, with respective index sets $I_{\mathrm{train}}$ and $I_{\mathrm{cal}}$, such that $|I_{\mathrm{cal}}| = n := c\cdot N$ and $|I_{\mathrm{train}}| = m := (1-c) \cdot N$ for a proportion $c \in (0,1)$.
In the rest of the paper, we consider the asymptotic scenario where $n \to \infty$ and $n/N = c$ for a fixed proportion $c \in (0,1)$. Given a user-specified miscoverage level $\alpha$, we aim to construct a sequence of \textit{asymptotically probably approximately correct (APAC)} \citep{qiu2024efficient}, or, equivalently, asymptotically \textit{training-set conditionally valid} \citep{vovk2012conditional} lower prediction bounds $(\hat{L}_n)_{n\geq 1}$ such that, with miscoverage level $\alpha$ and confidence level $1 - \beta$ ($\alpha, \beta \in (0, 1)$), 
\begin{equation}\label{eq: apac}
    \mathrm{Pr}_{\mathcal{D}} \left( \mathrm{Pr}_{(W,T) \sim P_0} \left( T > \hat{L}_n(W) \mid \mathcal{D}\right) \geq 1- \alpha \right) \geq 1 - \beta - o(1)
    \text{ as } n \rightarrow \infty,
\end{equation}
where $(W,T)$ is from a future full observation independently drawn from $P_0$, and the $o(1)$ term tends to zero as $n \rightarrow \infty$. We will provide a more detailed discussion of the difference between this guarantee and the aforementioned PAAC guarantee in Section~\ref{sec:7}.

Our method can also construct LPBs with an asymptotic marginal guarantee, with exceptional robustness. A detailed definition of the different types of guarantees is provided in Appendix~\ref{def:guarantee}.

\section{Methodology}\label{sec:3}
With a given LPB, we treat the coverage as a mapping from the true distribution to a real-valued parameter.
In Section~\ref{sec:iden diff}, we show that the coverage can be estimated from the observed data despite censoring, and that this mapping is differentiable in an appropriate sense under a nonparametric model.
In Sections~\ref{sec: nuisance} and \ref{sec:one-step}, we construct an efficient estimator of the coverage based on the aforementioned differentiability. This estimator relies on a one-step correction akin to a single step of Newton’s method, based on the linear approximation of the true coverage from an initial estimator \citep{fisher2021visually}.

\subsection{Identification and pathwise differentiability} \label{sec:iden diff}

We first 
estimate the marginal coverage \(\Psi(P_{0};L) := P_{0}(T > L(W))\) for a covariate-dependent LPB $L$ of the survival time $T$, where $L$ may be estimated from the training data $\mathcal{D}_{\mathrm{train}}$ and $(T,W)$ is randomly drawn from $P_{0}$ independent from the data.
We may drop $L$ from the notation for conciseness when the LPB $L$ is clear from the context.
We denote summaries of $P_0$ with the subscript 0; for example,  
$E_0[f(O)] := E_{P_0} [f(O)]$.

We first provide a formal identification result for our parameter $\Psi(P_{0})$ for the LPB $L$ as a functional of the observed data distribution implied by $P_{0}$.
Specifically, we assume that there exists $t_0 \in (0,\infty)$ such that $L(W) \leq t_0$ almost surely and the following conditions hold: 

\begin{condition}[Conditional independence]\label{id: cond1}
    \( T \cdot\mathbbm{1}(T \leq t_0) \perp\!\!\!\perp C \cdot\mathbbm{1}(C \leq t_0) \mid W \),
\end{condition}

\begin{condition}[positive probability]\label{id: cond2}
    \( P_0(C \geq t_0 \mid W=w) > 0\) for $P_0$-almost every $w$.
\end{condition}

Condition~\ref{id: cond1} allows the survival time and censoring time to be dependent, as long as they are conditionally independent given the covariate $W$. Additionally, Condition~\ref{id: cond2} ensures that within each stratum determined by the value of covariate $W$, there exists a nonzero probability of remaining uncensored at time $t_0$. This assumption is reasonable because, otherwise, for some stratum, all censorings occur before time $t_0$, and thus the observed data do not contain enough information about the distribution of $T \cdot \mathbbm{1}(T \leq t_0)$.

Under the two conditions, with the LPB $L$ fixed, $\Psi(P_0)$ depends on $P_0$ only through the marginal distribution of $W$ and the conditional survival function $S_0:(u \mid w) \mapsto \mathrm{Pr}_{P_0}(T > u \mid W=w)$. In this way, we can equivalently view $\Psi$ as a functional of the covariate distribution and the conditional survival function. 
It is well known that the conditional survival function $S_0$, the conditional censoring function $G_0$, and $\Psi(P_0;L)=E_0[S_0(L(W) \mid W)]$ are all estimable under these two conditions \citep[e.g.,][]{fleming2000survival}.

Then, we present the nonparametric efficient influence function (EIF) of $\Psi$. It characterizes the smallest large-sample variance among all \emph{regular estimators} and provides a basis for constructing an efficient estimator that attains this optimal variance.
For any LPB $L$, any generic survival functions $G$ and $S$, with $\Lambda$ denoting the cumulative hazard function associated with $S$ which can be expressed as an integral involving terms of $S$ (see Eq.~\eqref{eq:lambda} in Proposition~\ref{prop}), we define the function
\begin{align*}
    \phi(S, G; L) : (w,\delta,y) &\mapsto S(L(w) \mid w) \\ &\qquad\times \left[ 1 - \left\{ \frac{I(y \leq L(w) , \delta = 1)}{S(y \mid w)G(y \mid w)} - \int_{(0,L(w)\wedge y]} \frac{\Lambda(du \mid w)}{ S(u \mid w)G(u \mid w)} \right\} \right]
\end{align*}
whenever the denominators are nonzero.
With $P$ denoting a distribution of the tuple $O=(W,\Delta,Y)$, we also define the function
\begin{equation}
    D(P, G, S; L): o:=(w,\delta,y) \mapsto \phi(S,G;L)(o) - \Psi(P;L). \label{eq: d_tau}
\end{equation}

\begin{theorem}\label{thm: eif}
If there exists \(\eta_1 > 0\) such that 
$G_0(L(w) \mid w) \geq \eta_1$
for \(P_0\)-almost every \(w\) such that \(S_0(L(w) \mid w) > 0\), then \(P \mapsto \Psi(P;L) \) is a pathwise differentiable parameter in a nonparametric model with efficient influence function $D(P_0, G_0, S_0; L)$ at $P_0$.
\end{theorem}

We provide the proof of Theorem~\ref{thm: eif} in Appendix~\ref{app:eif} . 

\subsection{Estimation of nuisance functions}\label{sec: nuisance}

Following Theorem~\ref{thm: eif} and Eq.~\eqref{eq:lambda}, we see the derived EIF $D(P_0, G_0, S_0; L)$ involves two nuisance functions $S_0$ and $G_0$.
For any distribution $P$ of the observed data point $O$, define
\begin{align}\label{eq: nuisance}
     G_P(t \mid w) &:= \mathrm{Pr}_P(C > t \mid W = w), \nonumber \\
    S_P(t \mid w) &:= \mathrm{Pr}_P(T > t \mid W = w).
\end{align}
Under Conditions~\ref{id: cond1}--\ref{id: cond2}, we can estimate these nuisance functions with many existing methods, including the Weibull regression \citep{Zhang2016}, Cox proportional hazards model \citep{cox1972regression}, and random forests \citep{Ishwaran2008}.

Let $S_n$ and $G_n$ be estimators of the nuisance functions $S_0$ and $G_0$, respectively, using only the observations from the training set $\{O_i: i \in I_{\mathrm{train}}\}$. Next, we discuss the parameterization of the LPB function $L$. Primarily, the choice of $L$ should approximate the ``oracle" $\alpha$-quantile of $ T \mid W = w $. Thus we parameterize the LPB function as the $\tau$-quantile of the conditional survival time. Moreover, the data must contain enough information about $\mathrm{Pr}(T > L(w))$, that is, the probability of censoring by $L(w)$ must not be too small.
Thus, we parameterize the LPB function as 
\begin{equation}
     L_{n,\tau}(w) :=L_{S_n,G_n,\tau}(w) := \min( S_{n}^{-1}(1-\tau\mid w), G_{n}^{-1}(\eta_2\mid w)),\label{eq:L_tau}
\end{equation}
where $\eta_2$ is a user-specified number, similarly to \citet{gui2024conformalized}.
The regularization by $G_n^{-1}(\eta_2 \mid w)$ stabilizes estimation of coverage by ensuring Condition~\ref{id: cond2} to hold.
We take $\eta_2 = 1e^{-3}$ in our implementation.

\subsection{One-step estimator}\label{sec:one-step}
Next, we construct our one-step estimator for the coverage.
We let $P_n$ be the empirical distribution 
of the calibration data
$\{O_i: i \in I_{\mathrm{cal}}\}$.
Let $\hat{P}_n$ be the distribution with $W$ distributed as implied by the empirical distribution $P_n$, and the distribution of $(\Delta,Y) \mid W$ implied by $(G_n,S_n)$ and Condition~\ref{id: cond1}. Then,  
\begin{equation*}
    \Psi(\hat{P}_n;L_{n,\tau}) = n^{-1} \sum_{i=1}^n S_n(L_{n,\tau}(W_i) \mid W_i)
\end{equation*}
denote the plug-in estimator of $\Psi(P_0;L_{n,\tau})$.

Next, we describe our one-step corrected estimator for a given LPB $L_{n,\tau}$, when $L_{n,\tau}$ is independent of the calibration data. 
The plug-in estimator $\Psi_{\tau}(\hat{P}_n;L_{n,\tau})$ is suboptimal, even if $S_n$ is a good estimator of $S_0$. This is because its convergence rate is typically dictated by the convergence rate of $S_n$, which is often much slower than the optimal root-$n$ rate. The issue arises especially when $S_0$ is estimated with flexible machine learning methods.

We reduce this bias by applying a one-step correction, which adds the empirical average of the estimated EIF over the calibration data to $\Psi_{\tau}(\hat{P}_n;L_{n,\tau})$:
\begin{align}
    \hat{\psi}_{n,\tau} &= \Psi_{\tau}(\hat{P}_n;L_{n,\tau}) + \frac{1}{|I_{\mathrm{cal}}|} \sum_{i \in I_{\mathrm{cal}}} D(\hat{P}_n, G_{n}, S_{n}; L_{n,\tau})(O_i).\label{eq: est}
\end{align}
The algorithm of the estimator is explicitly described in Algorithm~\ref{alg}.

\begin{algorithm}[ht]
\caption{Split one-step estimator of coverage error $\Psi(P_0;L_{n,\tau})$}\label{alg}
\begin{algorithmic}[1]
\Require Data \( \mathcal{D} \), a tuning parameter $\tau \in (0,1)$ defining the LPB $L_{n,\tau}$ in Eq.~\eqref{eq:L_tau}, algorithms to estimate nuisance survival functions $S_0$ and $G_0$ in \eqref{eq: nuisance}, constant $\eta_2$ in \eqref{eq:L_tau}.
\State Partition the data \( \mathcal{D} \) into a training set \( \mathcal{D}_\mathrm{train} \) and a calibration set \( \mathcal{D}_\mathrm{cal} \), with respective index sets \( I_\mathrm{train} \) and \( I_\mathrm{cal}\).
\State Estimate \( S_0 \) by \( S_{n} \), estimate \( G_0 \) by \( G_{n} \) both using data out of training set $\mathcal{D}_\mathrm{train}$.
\State Obtain the split one-step corrected estimator for LPB \( L_{n,\tau} \), on the calibration set $\mathcal{D}_\mathrm{cal}$:
$$\hat{\psi}_{n,\tau} = \frac{1}{|I_{\mathrm{cal}}|} \sum_{i \in I_{\mathrm{cal}}} \phi(S_n, G_n; L_{n,\tau})(O_i).$$
\end{algorithmic}
\end{algorithm}

\subsection{Confidence Lower Bound and the Selection of Lower Prediction Bound}

Suppose a finite set $\mathcal{T}_n$ of candidate tuning parameters $\tau$ is specified by the user, possibly depending on the sample size. We aim to select $\hat{\tau}_n \in \mathcal{T}_n$ such that the corresponding LPB $L_{n,\hat{\tau}_n}$ is APAC. For each $\tau \in \mathcal{T}_n$, we obtain the asymptotically efficient estimator $\hat{\psi}_{n, \tau}$ of the coverage $\Psi(P_0;L_{n,\tau})$ and construct a Wald confidence lower bound (CLB). We estimate the asymptotic variance $\sigma^2_{0,n,\tau} := E_0[D(P_0, G_0, S_0; L_{n,\tau})(O)^2]$ with the plug-in estimator 
\[
\hat{\sigma}_{n,\tau}^2 := n^{-1} \sum_{i \in I_\mathrm{cal}} D(\hat{P}_n, G_n, S_n; L_{n,\tau})(O_i)^2
\]
via sample splitting. The $(1 - \beta)$ Wald-CLB is $\hat{\psi}_{n,\tau} - z_{\beta} \hat{\sigma}_{n,\tau}/\sqrt{n}$, where $z_\beta$ denotes the $(1 - \beta)$-quantile of the standard normal distribution.
 We then select $\tau$ based on the Wald-CLB:
\begin{equation*}
    \hat{\tau}_n := \max \left\{ \tau \in \mathcal{T}_n : \hat{\psi}_{n,\tau'} - z_\beta \hat{\sigma}_{n,\tau'}/\sqrt{n} \geq 1- \alpha \text{ for all } \tau' \in \mathcal{T}_n \text{ such that } \tau' \leq \tau \right\}.
\end{equation*}

\section{Theoretical guarantee}\label{sec:4}
We present the following sufficient conditions for ensuring the asymptotic properties of the proposed one-step estimator in Algorithm~\ref{alg}.\footnote{Throughout this paper, whenever random functions such as $S_n$ and $G_n$ appear in expectations, the expectations average only over $(W,\Delta,Y) \sim P_0$ with the random functions treated as fixed. Thus, terms such as $E_0 [\sup_{u\in[0,t]} |1/{G_{n}(u|W)}-1/G_0(u|W)|]^2$ are still random through the random nuisance estimators $S_n$ and $G_n$.}

\begin{enumerate}[label=A\arabic*]
     \item Limit 
    of nuisance estimators: There exist fixed survival functions $G_\infty$ and $S_\infty$, such that $P_0(S_\infty(t_0 \mid W) \geq \zeta)=1$ for some constant $\zeta>0$. And, for any $t \in [0,t_0)$, 
    \label{cond: C1}
\begin{align*}
E_0&\left[\sup\limits_{u\in[0,t]}\left|\frac{1}{G_{n}(u|W)}-\frac{1}{G_{\infty}(u|W)}\right|^2\right]=o_p(1), \\
E_0&\left[\sup\limits_{u\in[0,t]}\left|\frac{S_{n}(t|W)}{S_{n}(u|W)}-\frac{S_{\infty}(t|W)}{S_{\infty}(u|W)}\right|^2\right]=o_p(1).
\end{align*}
\item There exists $\eta > 1$ such that, with probability tending to one, for $P_0$-almost every $w$, $G_n(t \mid w) \geq 1/\eta$. \label{cond: C2} 
\item $E_0\left[\sup_{u\in[0,t]}\sup_{z\in[0,u]}\left|\frac{S_{n}(u|W)}{S_{n}(z|W)}-\frac{S_{\infty}(u|W)}{S_{\infty}(z|W)}\right|^2\right]=o_p(1)$. \label{cond: C3} 
\item Sufficient rate of nuisance estimators: it holds that $\sup_{\tau\in [0, 1]} r_{n,\tau} = o_p(n^{-1/2})$, where 
    $$r_{n,\tau} :=  E_0 \left| S_{n}(L_{n, \tau}(w) | W) \int_{(0,L_{n,\tau}(w)]} \left\{ \frac{G_0(u | W)}{ G_{n}(u | W)} - 1 \right\} \left( \frac{S_0 }{S_{n}}-1 \right)(du |  W) \right|.$$
  \label{cond: C4}

    \item Consistency of nuisance estimators: $S_\infty=S_0$ and $G_\infty=G_0$ for $S_\infty$ and $G_\infty$ in Condition~\ref{cond: C1}. \label{cond: C5}
\end{enumerate}
The remainder term $r_{n,\tau}$ in Condition~\ref{cond: C4} is not a mixed bias \citep{rotnitzky2021characterization} but a \emph{cross integrated error term} \citep{Ying2023}. It is a common assumption that a remainder of this form is $o_p(n^{-1/2})$ in nonparametric survival problems with right-censoring \citep[e.g.,][]{westling2023inference,Ying2023,Wang2024}.

These conditions lead to our following main result on the one-step estimator $\hat{\psi}_{n,\tau}$.

\begin{theorem}
[Asymptotic efficiency of one-step corrected estimator] \label{thm3}
Under Con-\\ditions~\ref{cond: C1}--\ref{cond: C5}, with the one-step corrected estimator $\hat{\psi}_{n,\tau}$ from Eq.~\eqref{eq: est}, the coverage $\Psi(P_0;\allowbreak L_{n,\tau}) := E[S_0(L_{n,\tau}(W) \mid W)$,
the gradient $D(P_0, G_0, S_0; L_{n,\tau})$ from Eq.~\eqref{eq: d_tau}, the conditional censoring function $G_0$ and the conditional survival function $S_0$ from Eq.~\eqref{eq: nuisance}, we have
\begin{equation}
    \sup_{\tau \in [0,1]} \left| \hat{\psi}_{n,\tau} - \Psi(P_0; L_{n,\tau}) - \frac{1}{|I_{\mathrm{cal}}|} \sum_{i\in I_{\mathrm{cal}}} D(P_0, G_0, S_0; L_{n,\tau})(O_i) \right| = o_p(n^{-1/2}). \label{eq:ovvar}
\end{equation}
\end{theorem}
The proof of Theorem~\ref{thm3} can be found in Appendix~\ref{app:thm}.

Note that $E_0[D(P_0, G_0, S_0; L_{n,\tau})(O_i)]=0$, by the boundedness of $D(P_0, G_0, S_0; L_{n,\tau})$, $\sigma_{0,n,\tau}^2 := E_0[D(P_0, G_0, S_0; L_{n,\tau})(O_i)^2]$ is finite. Since $L_{n,\tau}$ is independent of the calibration data, Theorem~\ref{thm3} shows that, conditioning on the training data, $\sqrt{n} (\Psi_{\tau}(\hat{P}_n;L_{n,\tau}) - \Psi(P_0; L_{n,\tau}))/\sigma_{0,n,\tau}$ converges in distribution to the standard normal distribution as $n \to \infty$.
Next, because of the consistency in Condition~\ref{cond: C1} and Theorem~\ref{thm2}, the plug-in estimator $\hat{\sigma}_{n,\tau}^2$ is consistent for $\sigma_{0,n,\tau}^2$.
Because of the aforementioned asymptotic normality and the consistency of $\hat{\sigma}_{n,\tau}^2$, we have that
$$\left| \mathrm{Pr}(\Psi(P_0; L_{n,\tau}) \geq \hat{\psi}_{n,\tau} - z_{\beta} \hat{\sigma}_{n,\tau} / \sqrt{n}) - (1 - \beta) \right| \to 0 \quad (n \to \infty)$$
for every $\tau \in \mathcal{T}_n$ under the conditions in Theorem~\ref{thm3}. 
By Theorem~6 in \citet{qiu2023prediction}, the LPB $L_{n,\hat{\tau}_n}$ with the selected tuning parameter $\hat{\tau}_n$ is APAC.

\section{Simulations}\label{sec:5}

In this section, we empirically verify the validity of our method using simulated data. First, we provide detailed setups for the simulations, including the parameters of data distributions, baselines, and metrics. Next, we present both the quantitative and qualitative results for training-set conditional validity and marginal validity.

\subsection{Setups}

\paragraph{Setups.} We conduct experiments in six synthetic settings based on a well-established baseline  \citep{gui2024conformalized}. Each experiment generates multiple i.i.d. datasets, split equally into training, calibration, and test sets. Synthetic datasets 
are generated from a common template:
\( P_W \sim \text{Unif}([0, 4]^p) \), 
\( P_{T | W} \sim \mathrm{LogNormal}(\mu(W), \sigma^2(W)) \),
where \( p \) is the covariate dimension. We set $\alpha = 0.1$ and $\beta = 0.05$.

Settings 1-2 are univariate with independent censoring, while Settings 3-4 involve \allowbreak covariate-dependent censoring. Settings 5-6 are the most challenging, with high-dimensional covariates (\( p = 10 \)) and covariate-dependent censoring. Table~\ref{tbl:setup} summarizes the parameters for each setup. 
\begin{table}[h]
\centering
\bgroup
\setlength\tabcolsep{1pt}
\begin{tabular}{c c c c c}
\hline
Setting & $p$ & $\mu(w)$ & $\sigma(w)$ & $P_{C | W}$ \\
\hline
1 & 1 & $0.632x$ & 2 & $\operatorname{Exp}(0.1)$ \\[2ex]
2 & 1 & $3 \cdot \mathbbm{1}\{w > 2\} + w \cdot \mathbbm{1}\{w \leq 2\}$ & 0.5 & $\operatorname{Exp}(0.1)$ \\[2ex]
3 & 1 & $2 \cdot \mathbbm{1}\{w > 2\} + w \cdot \mathbbm{1}\{w \leq 2\}$ & 0.5 & $\operatorname{Exp}\left(0.25 + \dfrac{6 + w}{100}\right)$ \\[2ex]
4 & 1 & $3 \cdot \mathbbm{1}\{w > 2\} + 1.5x \cdot \mathbbm{1}\{w \leq 2\}$ & 0.5 & $\operatorname{LogNormal}\left(2 + \dfrac{2 - w}{50},\ 0.5\right)$ \\[2ex]
5 & 10 & $0.126(w_1 + \sqrt{w_3 w_5}) + 1$ & 1 & $\operatorname{Exp}\left(\dfrac{w_{10}}{10} + \dfrac{1}{20}\right)$ \\[2ex]
6 & 10 & $0.126(w_1 + \sqrt{w_3 w_5}) + 1$ & $\dfrac{w_2 + 2}{4}$ & $\operatorname{Exp}\left(\dfrac{w_{10}}{10} + \dfrac{1}{20}\right)$ \\
\hline
\end{tabular}
\caption{Parameters used in the six experimental settings.}\label{tbl:setup}
\egroup
\end{table}

\paragraph{Metrics.} For each run, we compute the proportion of runs where the coverage on the test set is at least $1-\alpha=0.9$.
For the methods achieving marginal coverage, we compute the empirical coverage and the average LPB on the test set. 
\begin{align*}
\text{Empirical coverage} &= \frac{1}{|I_\text{test}|} \sum_{i \in I_\text{test}} \mathbbm{1}\left\{T_i > \hat{L}(W_i) \right\}, \\
\text{Average LPB} &= \frac{1}{|I_\text{test}|} \sum_{i \in I_\text{test}} \hat{L}(W_i).
\end{align*}

\subsection{Training-set conditional validity results}


First, we verify our main theorem on the training-set conditional guarantee. For all the experiments, we fit \(S_n\) and \(G_n\) using the \texttt{survSuperLearner} R package \citep{westling2023inference} with the Cox proportional hazards model \citep{cox1972regression}, the Weibull regression model \citep{Zhang2016}, the generalized additive model \citep{hastie1990generalized}, and the random forest \citep{Ishwaran2008} in the library. For all univariate settings, we consider sample sizes $n = 200$, $500$, and $1000$. For multivariate settings, we also consider larger sample sizes $n = 4000$ and $8000$, and include a screening algorithm in \texttt{survSuperlearner} to obtain sparse estimators $S_n$ and $G_n$.
For each sample size, we run all methods on $200$ randomly generated datasets.

\begin{figure}[h!]
    \centering
    \begin{subfigure}{\linewidth}
        \centering
        \includegraphics[width=0.9\linewidth]{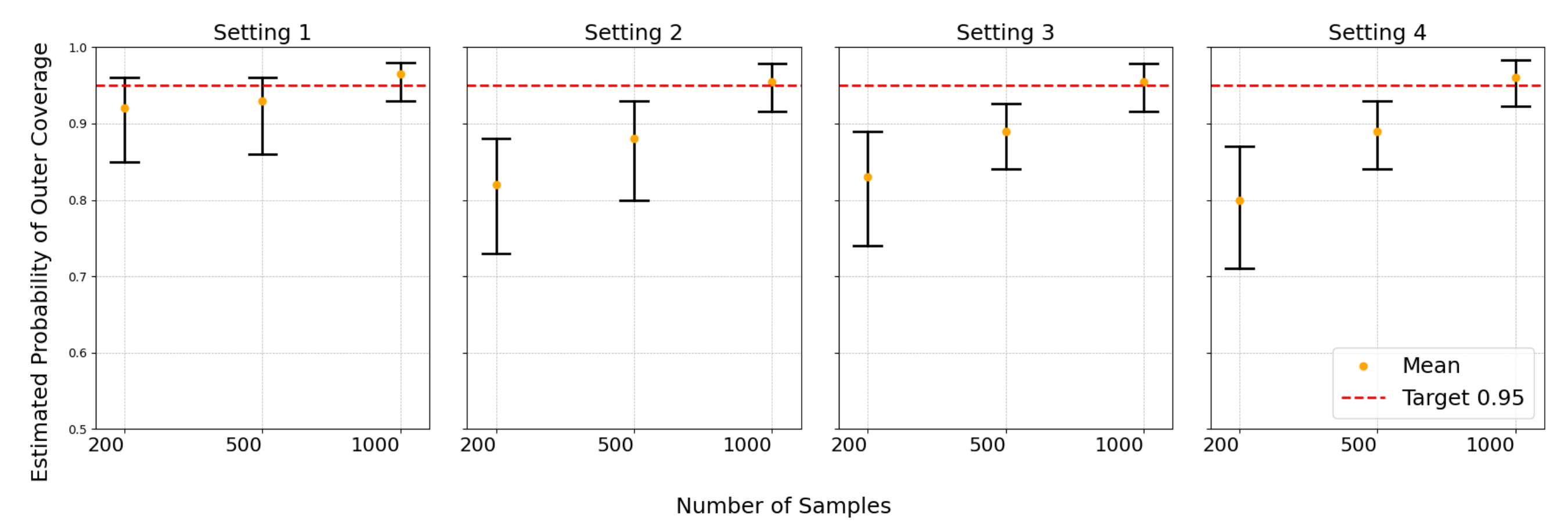}
    \end{subfigure}
    
    \vspace{0.5em} 

    \begin{subfigure}{\linewidth}
        \centering
        \includegraphics[width=0.9\linewidth]{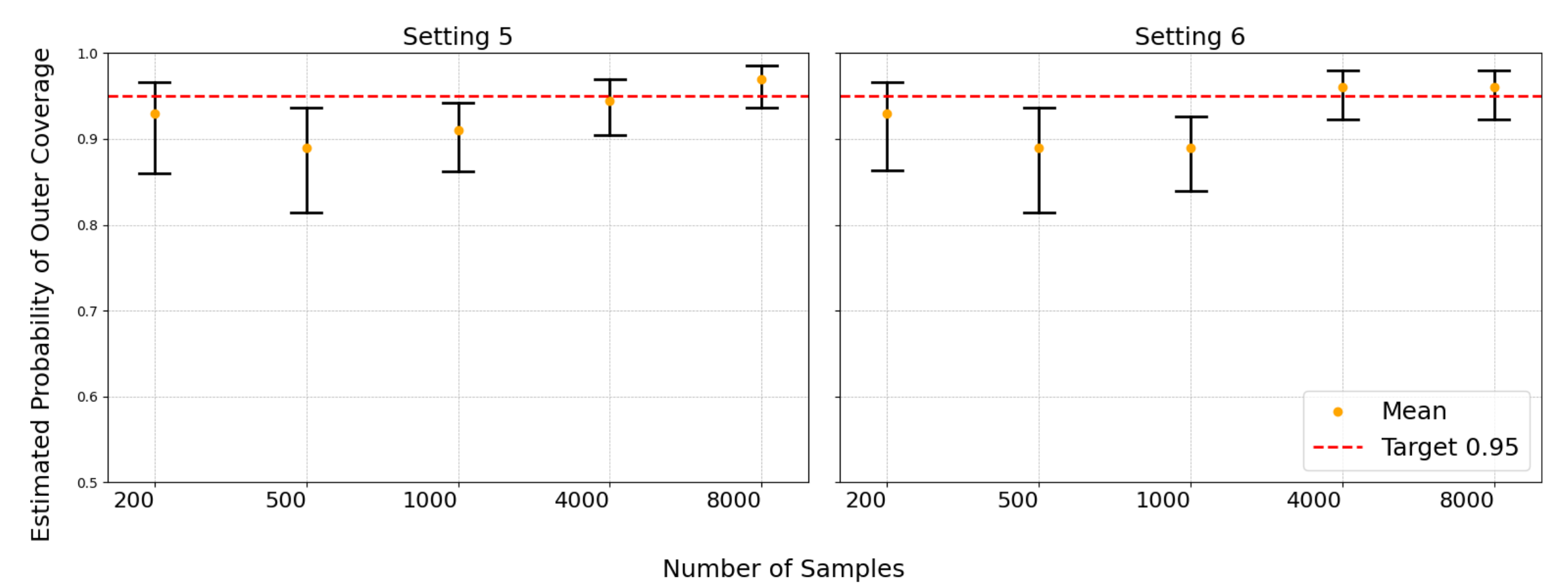}
    \end{subfigure}
    \caption{Empirical proportion of simulations where the estimated coverage $\mathrm{Pr}_{(W,T) \sim P_0} ( T > \hat{L}_n(W))$ is at least $\alpha$, along with a 95\% Wilson score confidence interval for this proportion. The red horizontal dashed line represents the desired confidence level, namely the proportion $1 - \beta = 0.95$.}
    \label{fig:cond}
\end{figure}

The empirical proportion of coverage being at least $1-\alpha$ is presented in Figure~\ref{fig:cond}. Despite some undercoverage in small-sample and high-dimensional cases, the empirical proportion for TCsurv steadily increases as the sample size grows, eventually reaching the desired level. This result 
aligns with
the asymptotic PAC validity of our method.

\subsection{Marginal coverage experiments}

Although our primary goal is to construct an asymptotic PAC prediction set, our method can also achieve asymptotic marginal coverage with a simple modification in the tuning parameter selection. The modified method and corresponding results are provided in Appendix~\ref{app:marg}, with its asymptotic marginal coverage following from the consistency property stated in Theorem~\ref{thm2}. Notably, our marginal performance demonstrates a strong double robustness property, in the sense that it achieves asymptotic marginal coverage even when one nuisance function is mis-specified. 

Next, we describe the experiments for verifying the asymptotic marginal coverage guarantee in Definition~1.1 of \citet{farina2025doubly}, which is also stated in Eq.~\eqref{eq:asymp_marg} in Appendix~\ref{def:guarantee}. For each setting, we generate \( 100 \) i.i.d. datasets. We compare our method with alternative existing distribution-free LPB methods for \(T\) (DFT):
\begin{itemize}
    \item \textbf{DFT-fixed} \citep{candes2023conformalized}: Conformalized LPB that constructs $\hat{L}$
    from calibration samples $(W_i, Y_i \land c_0)$ among those with $C_i \geq c_0$ such that $\mathrm{Pr}(T \land c_0 \geq \hat{L}(W)) \geq 1 - \alpha$. Here, $c_0$ is a threshold tuned via grid search: $c_0 = \arg\max_{c_0 \in C} n^{-1} \sum_{i \in I_{\mathrm{cal}}} \hat{L}(W_i)$. 
    
    \item \textbf{DFT-adaptive-T} \citep{gui2024conformalized}: Adaptive conformalized LPB. The candidate LPB is the estimated $\tau$-th conditional quantile of \(T\mid W = w\), given by $\hat{L}(w) = S_n^{-1}(1-\tau\mid w)$.
    
    \item \textbf{DFT-adaptive-CT} \citep{gui2024conformalized}: Adaptive conformalized LPB, where the candidate LPB is $\hat{L}(w) = \min( S_n^{-1}(1-\tau\mid w), G_n^{-1}(1/\log(n)\mid w))$.
\end{itemize}

\begin{figure}[t]
    \centering
    \begin{subfigure}{\linewidth}
        \centering
        \includegraphics[width=0.95\linewidth]{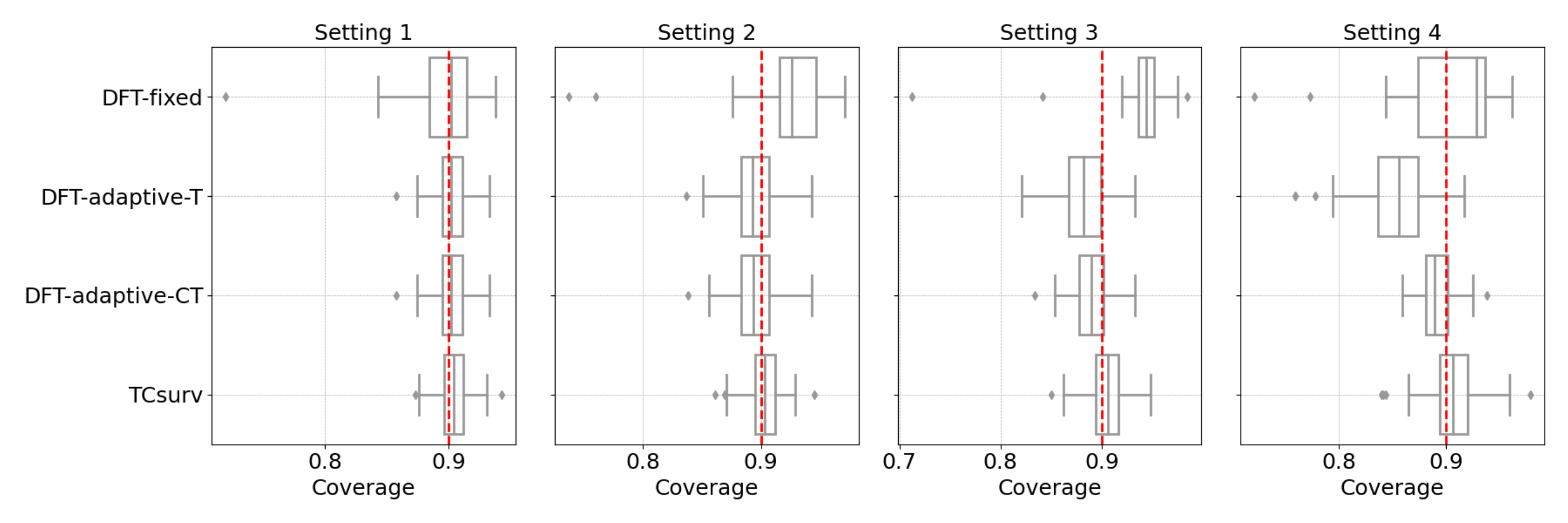}
    \end{subfigure}
    
    \vspace{1em} 

    \begin{subfigure}{\linewidth}
        \centering
        \includegraphics[width=0.95\linewidth]{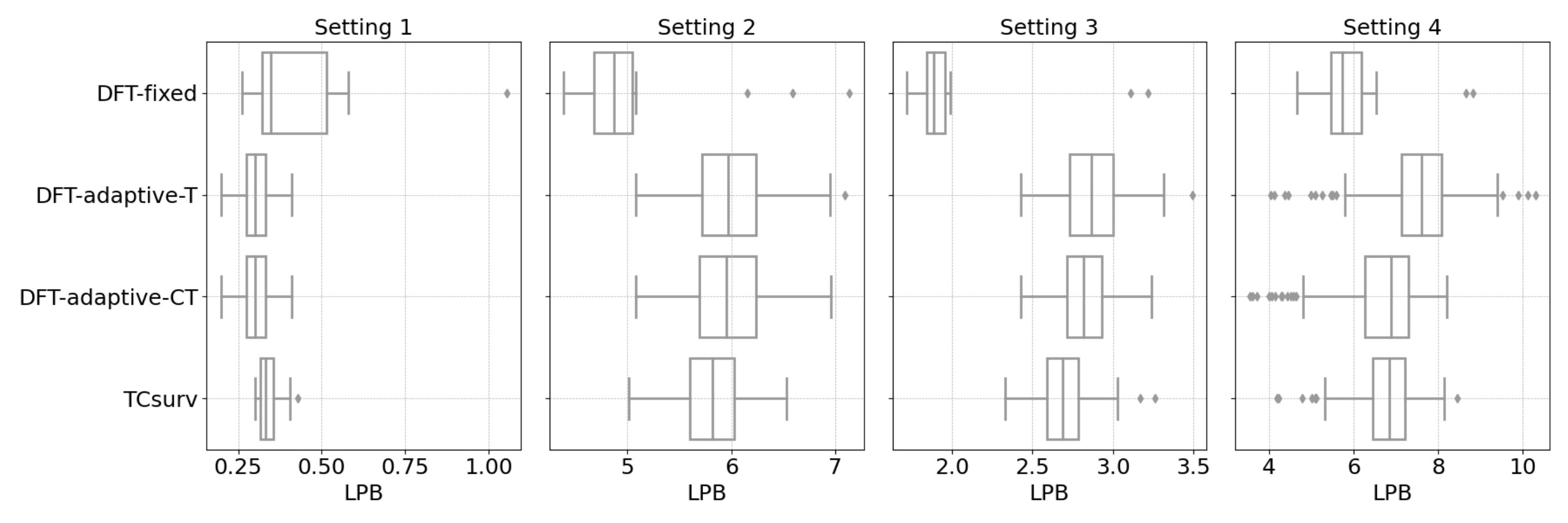}
    \end{subfigure}
    
    \caption{Empirical coverage (top) and average LPBs (bottom) of all candidate methods under settings 1–4, where $W$ is univariate. The boxplots show results from 100 independent draws of datasets. The dashed red line indicates the target coverage level of $1 - \alpha = 90\%$.
    }
    \label{fig:marg1}
\end{figure}

In all experiments, the algorithm for fitting $S_n$ is estimated using the \texttt{survSuperlearner} R package, while $G_n$ is estimated with the Cox proportional hazards model. Notably, all baseline methods require full access to the actual censoring time $C$, whereas our method does not. For a fair comparison, we present results where our method also has access to $C$ when estimating the censoring distribution. Additionally, as shown in Appendix~\ref{app:res2}, our method produces similarly efficient LPBs even when $C$ is not fully observed when $T$ precedes $C$. 

We also conducted experiments where $G_n$ was estimated using a Gaussian process (GP), with the results provided in Appendix~\ref{app:res1}. Notably, the GP is misspecified, as the learned mean and variance exhibit negligible variation across subgroups with different covariate values. The results show that, in this case, we still obtain robust coverage and tight LPBs.

\paragraph{Results.} 
We present the results of empirical coverage and average LPBs for all candidate methods under univariate settings in Figure~\ref{fig:marg1} and multivariate settings in Figure~\ref{fig:marg2}.
Across all settings, our method exhibits the most stable coverage performance at the target level of $1-\alpha$. In contrast, DFT-fixed tends to overcover due to the grid scale used for selecting the threshold $c_0$. While DFT-adaptive-T and DFT-adaptive-CT achieve the desired coverage level in the two homoscedastic univariate settings (Settings~1 and 2), their coverage deteriorates significantly in Settings~3 and 4.

In terms of average LPB, 
TCsurv is comparable to the other methods. Notably, in the most complex univariate Setting 4, 
our method is one of the only two methods that successfully maintain the desired coverage level, and it achieves a higher mean LPB, making it more efficient.
This is because, in these settings, the censoring distribution depends on the covariate in a complex manner and is poorly estimated.
TCsurv is doubly robust against poor estimation of one nuisance function, but the baseline methods are less robust.
Overall, our method is the only approach that consistently achieves desired coverage level while maintaining tight LPBs, demonstrating both its double robustness and efficiency.
\begin{figure}[h!]
    \centering
    \begin{subfigure}[b]{0.49\textwidth}
        \centering
        \includegraphics[width=\textwidth]{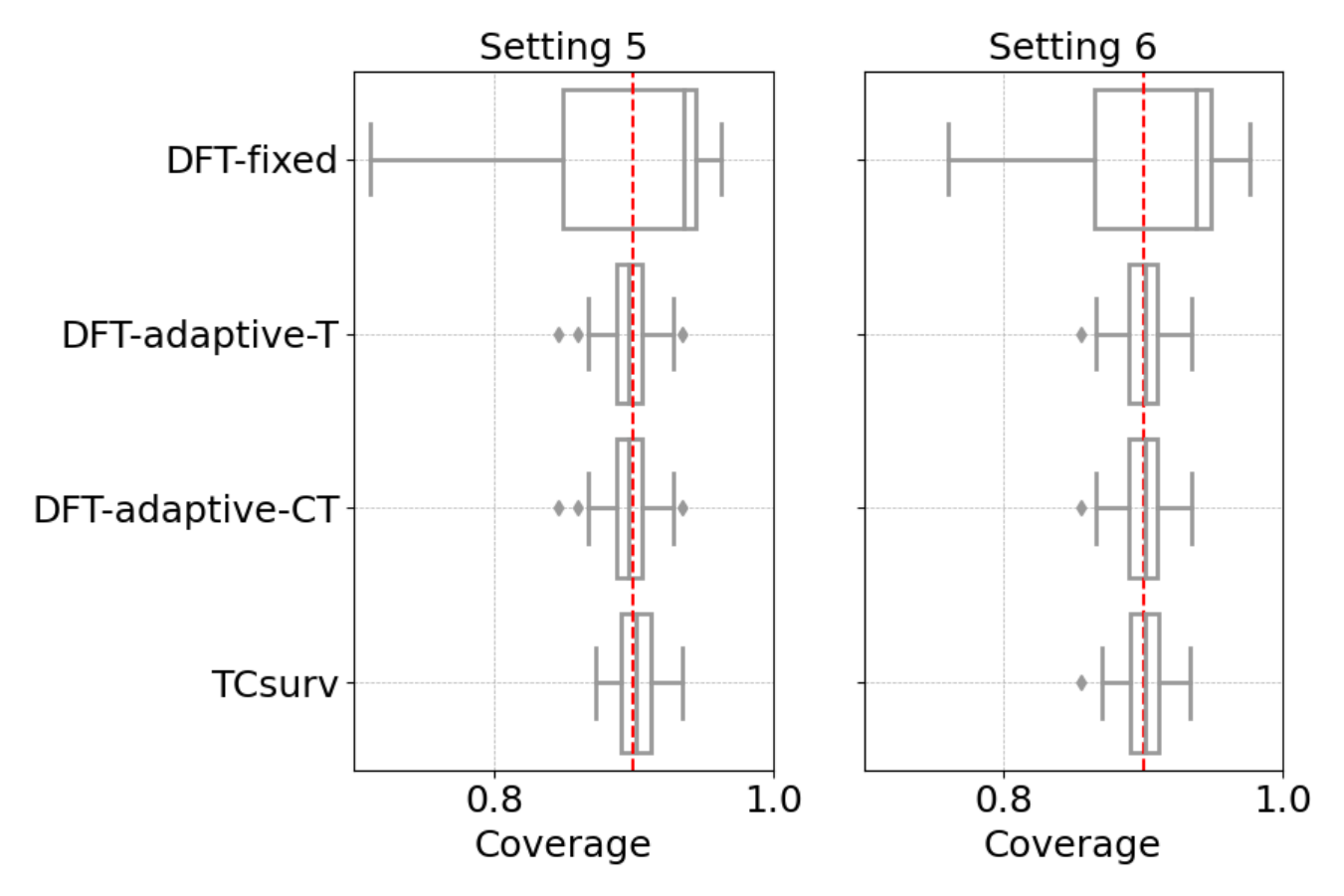}
    \end{subfigure}
    \hfill
    \begin{subfigure}[b]{0.49\textwidth}
        \centering
        \includegraphics[width=\textwidth]{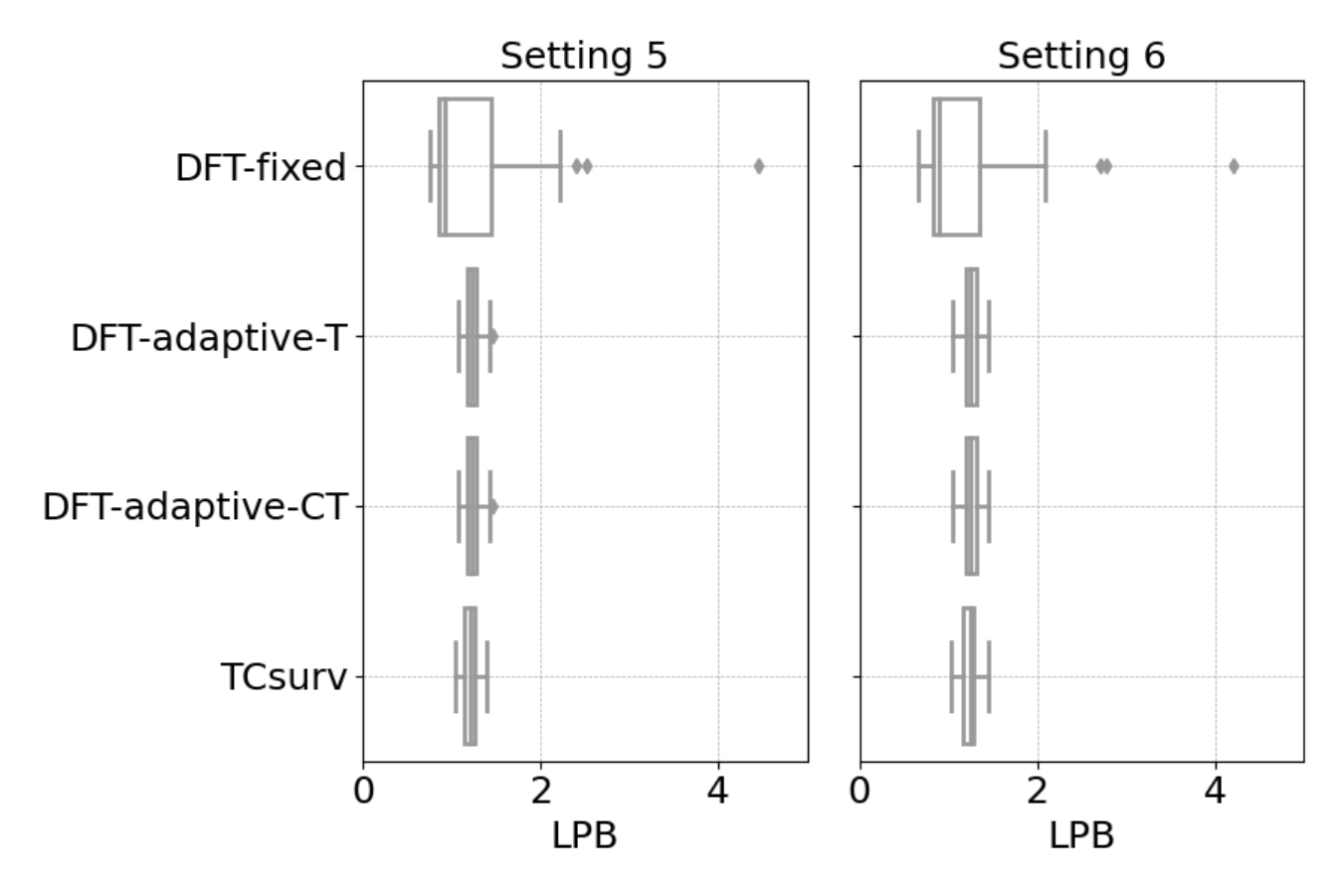}
    \end{subfigure}
    \caption{Boxplots of empirical coverage (left) and average LPBs (right) in the multivariate experimental settings. All other details are the same as in Figure~\ref{fig:marg1}.
    }
    \label{fig:marg2}
\end{figure}

\section{Application to mobile app data}\label{sec:6}

We apply our proposed method to a real-world mobile app dataset recording users' login activities.\footnote{Data downloaded from the \href{https://www.kaggle.com/datasets/bhuvanchennoju/mobile-usage-time-prediction}{Kaggle Users Active Time Prediction dataset}.} This dataset contains the timestamps of user activities for 2,476 users in a shared window of three weeks.
We follow the setup of \citet{gui2024conformalized} to predict the start time of a user's $K$-th active day, which is useful for targeted promotions. For example, companies may seek to determine the optimal timing for releasing advertisements to individual users, allowing them to promote their products while minimizing advertising costs.

Since the original data is right-censored, $T$ is not fully observed, leading to challenges for validation.
We therefore apply TCsurv to a further coarsened version of the data in our application to ensure that the original data can be used for validation.
We choose $K=9$ to ensure that 
the ground truth coverage rate in a held-out validation set can be computed.\footnote{That is, all users have at least $9$ active days during the experiments.}
The censored time $C$ is set to be $11$ or $12$ rather than those in the original data, depending on whether the user opened the app on the first day, and the data for constructing LPB is censored accordingly.
The resulting event rate is approximately $64\%$. 
We use three covariates for LPB: gender, age, and number of children. 

In our marginal experiments, we verify the guarantees using the same baselines as in the synthetic settings. We first apply a triple hold-out split, then sample with replacement within each part to obtain a training set, a calibration set, and a test set with the same size $n=1000$.

We then verify the training-set conditional guarantees using the same baselines as in the synthetic settings. We consider sample sizes of $n = 200$, $500$, and $1000$. The training-set conditional results are reported in Figure~\ref{fig:real} (right). Our method steadily improves as the sample size increases, eventually reaching the desired level. In conclusion, our method successfully achieves the asymptotic training-set conditional guarantee.

The marginal validity results are reported in Figure~\ref{fig:real} (left and middle). In this case, the difference between DFT-adaptive-T and DFT-adaptive-CT is negligible. All methods achieve 90\% coverage with similar average LPBs. Overall, our method achieves marginal coverage with a tight LPB.

\begin{figure}[h!]
    \centering
    \includegraphics[width=0.85\textwidth]{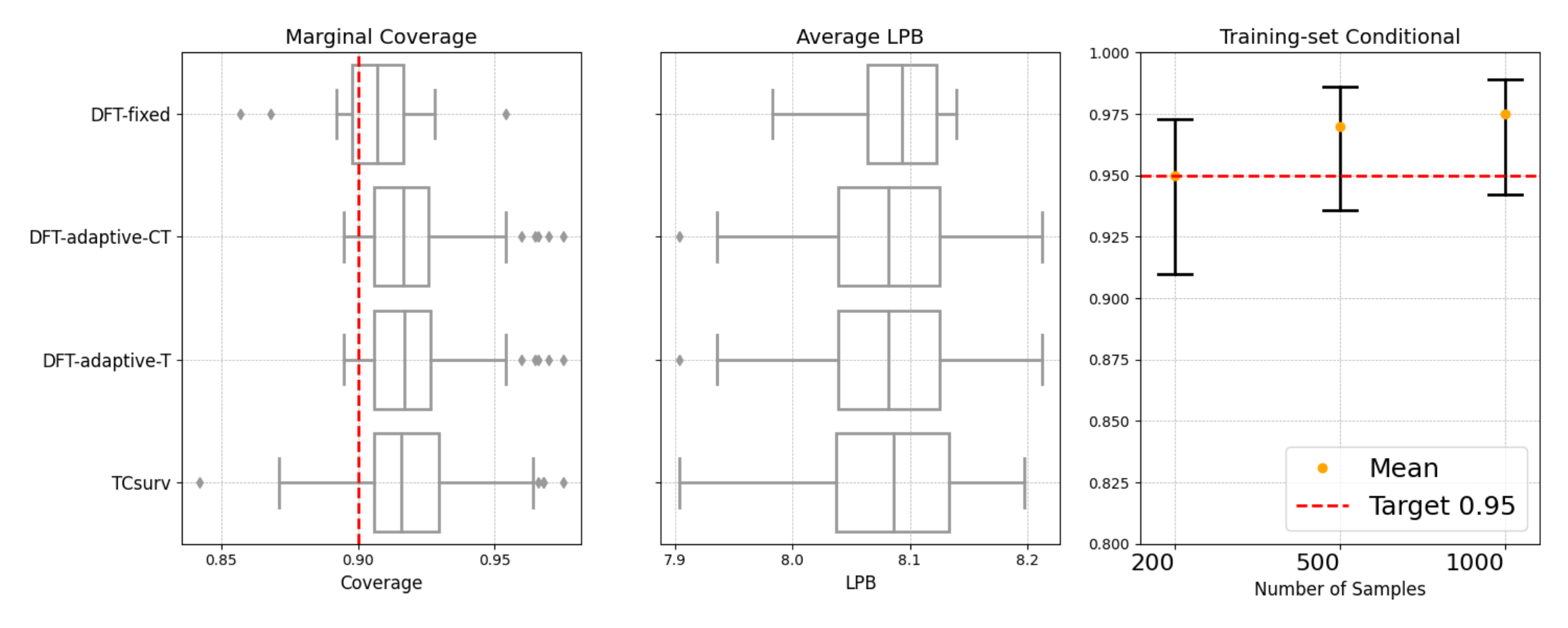}
    
    \caption{Empirical results on the mobile app activity dataset. 
Left: Empirical marginal coverage rate. 
Middle: Average LPBs for marginal performance. 
Right: Empirical proportion of runs where the estimated coverage error does not exceed $\alpha$, with a 95\% Wilson score confidence interval.}
    \label{fig:real}
\end{figure}

\section{Discussion}\label{sec:7} 

We develop a novel nonparametric method for constructing lower prediction bounds (LPBs) with large-sample training-set conditional coverage for survival times from right-censored data, extending recent methods designed for marginal coverage and Type I censoring. 
It achieves strong empirical performance, producing marginally and training-set conditionally valid LPBs that are both informative and robust compared to alternative methods. 
Because of the double robustness property, our method ensures consistency even in the presence of poor nuisance estimation, thereby enhancing its reliability in achieving marginal coverage in practical applications.

Many prior works establish PAC-variant guarantees that differ from ours, specifically the PAAC guarantee. \citet{gui2024conformalized} obtain this guarantee in the Type I censoring setting, while the concurrent work \citet{farina2025doubly} extends it to general right-censored data. The key distinction between APAC and PAAC lies in how they handle the asymptotically vanishing approximation error: APAC applies this approximation to the confidence level, whereas PAAC does so to the coverage error. APAC ensures that with a confidence level approaching to the target $1 - \beta$, the true coverage error remains within $\alpha$, though it may be slightly conservative. In contrast, PAAC guarantees that, with at least $1 - \beta$ confidence, the true coverage error does not significantly exceed $\alpha$, but small undercoverage may occur. 
APAC is particularly advantageous in scenarios where high-confidence control over coverage error is crucial, even at the cost of mild conservatism, making it especially valuable in safety-critical applications.

\citet{holmes2024two} recently proposed a method for also constructing upper prediction bounds. It is worth noting that our method can be conveniently modified to produce two-sided prediction intervals,
which we leave for future interest.
\acks{The authors thank Insup Lee for his helpful feedback and discussions on this work.
This work was supported by ARO W911NF-20-1-0080 and the NSF-NIH SCH grant.}


\bibliography{sample}

\newpage

\appendix
\section{Types of guarantees}\label{def:guarantee}
We say that $\hat{L}$ is a marginally calibrated lower prediction bound (LPB) \citep{candes2023conformalized, gui2024conformalized} at level $1 - \alpha$ if it satisfies 
\begin{equation*}
\mathrm{Pr}_{(W,T) \sim P_0} \Big(T \geq \hat{L}(W)\Big) \geq 1-\alpha.
\end{equation*}
Theorem~1 in \citet{candes2023conformalized} shows that finite-sample marginal calibration is rarely informative.
We thus consider an asymptotically 
marginally calibrated LPB.
We say that a sequence of LPB $\{\hat{L}_n\}_{n=1}^\infty$ is asymptotically 
marginally calibrated for $T$ at level $\alpha$
if \begin{equation}\label{eq:asymp_marg}
    \mathrm{Pr}_{(W,T) \sim P_0}(T \geq \hat{L}_n(W)) \geq 1 - \alpha - o_p(1) \text{ as } n \to \infty.
\end{equation}

In terms of training-set conditional coverage, \citet{gui2024conformalized, farina2025doubly} achieve an asymptotic PAC-variant guarantee other than APAC in Eq.~\eqref{eq: apac}:
\begin{equation}\label{eq:paac}
        \mathrm{Pr}_{\mathcal{D}} \left( \mathrm{Pr}_{(W,T) \sim P_0} \left(T \geq \hat{L}_n(W) \mid \mathcal{D} \right) \geq 1 - \alpha - o_p(1) \right) \geq 1 - \beta \text{ as } n \to \infty.
    \end{equation}
\citet{qiu2023prediction} called Eq.~\eqref{eq:paac} a probably asymptotically approximately correct (PAAC) guarantee in their Remark~3. 

\section{Asymptotic marginal validity}\label{app:marg}

In this section, we first provide a theorem establishing the consistency of our estimator $\Psi(\hat{P}_n; L_{n,\tau})$ under weaker conditions than those in Theorem~\ref{thm3}. Next, we present our method to select the threshold.

\subsection{Theorem}
We first introduce the following condition implied by Condition~\ref{cond: C5}.
Recall $S_\infty$ and $G_\infty$ defined in Condition~\ref{cond: C1}.
\begin{enumerate}[label=A\arabic*, start=6]
    \item Consistency of one nuisance estimator: $S_\infty=S_0$ \textit{or} $G_\infty=G_0$. 
\label{cond: C6}
\end{enumerate}
\begin{theorem}[Consistency]\label{thm2}
If Conditions~\ref{cond: C1}--\ref{cond: C2} and \ref{cond: C6}
hold, then for every $\tau \in (0, 1)$, $\Psi (\hat{P}_n;L_{n,\tau})\allowbreak \xrightarrow{P}\Psi (P_0;L_{n,\tau})$.
In addition, $\sup_{\tau \in [0,1]} |\Psi(\hat{P}_n; L_{n,\tau})- \Psi (P_0;L_{n,\tau})| = o_p(1)$.
\end{theorem}
We provide the proof in Section~\ref{app:thm}.

\subsection{Selection of threshold}\label{app:marg_threshold}

We select a threshold \( \hat{\tau}_n \) to ensure that the size of prediction sets is small and the estimated coverage is at least $1-\alpha$:
\begin{equation*}
    \bar{\tau}_n := \max \left\{ \tau \in \mathcal{T}_n : \hat{\psi}_{n,\tau'} \geq 1- \alpha, \quad \forall \tau' \in \mathcal{T}_n \text{ such that } \tau' \leq \tau \right\},
\end{equation*}
where $\mathcal{T}_n$ is a set of tuning parameter candidates specified by the user.
The asymptotic marginal validity of our method follows directly from the consistency property in Theorem~\ref{thm2}.
In our simulations and real data application, we set $\mathcal{T}_n$ to be a grid uniformly spaced within the interval $[0,1)$.

\section{Proof of theoretical results}\label{app:seca}

We use the following notation throughout the following sections.
For any distribution $P$ and any function $f$, let $P f := \int f d P$ and $\mathbbm{G}_n f := \sqrt{n} (P_n - P_0) f$ where $P_n$ denotes the empirical distribution in the calibration data.
In particular, $P_n f = n^{-1} \sum_{i \in I_{\mathrm{cal}}} f(O_i)$.
We use $\|f\|_{2,P_0}$ to denote the $L_2(P_0)$-norm of $f$, namely $(\int f^2 dP)^{1/2}$.

\subsection{Derivation of EIF}\label{app:eif}

The following proposition shows that, under Conditions~\ref{id: cond1} and \ref{id: cond2}, the survival function $S_0$ is estimable by showing that $S_0$ can be represented in terms of the observed right-censored data.
By switching the role of $T$ and $C$, we can show that $G_0$ is also estimable.
\begin{proposition}
\label{prop}
Denote $F_{0,1}(t \mid w) := P_0(Y \leq t, \Delta = 1 \mid W = w), \quad
R_0(t \mid w) := P_0(Y \geq t \mid W = w)$, and 
$\Lambda_0(t \mid a, w) := \int_{(0,t]} \frac{F_{0,1}(du \mid a, w)}{R_0(u \mid a, w)}$. If Conditions~\ref{id: cond1} and \ref{id: cond2} hold for some  \(t_0 \in (0, \infty)\), we have 
$$S_0(v \mid w) = \Prodi_{(0,v]} \left(1-\Lambda_0(du \mid w)\right)$$
where $\Prodi$ denotes the Riemann-Stieltjes product integral \citep{gill1990survey}.
\end{proposition}
\begin{proof}[Proof of Proposition \ref{prop}]
Recall
\[
S_0(L(w) \mid w) = P_0(T > L(w) \mid W = w).
\]
Since \(T\) is a positive random variable, by Theorem 11 of \citet{gill1990survey} we have
\begin{equation}\label{eq:s}
    S_0(L(w) \mid w) = P_0(T > L(w) \mid W = w) = \Prodi_{u \in (0,L(w)]} \left\{1 - \Lambda_{0}(du \mid w)\right\}.
\end{equation}
for 
\begin{equation}\label{eq:lambda_int}
    \Lambda_{0}(v \mid w) := -\int_{(0,v]} \frac{S_0(du \mid w)}{S_0(u^- \mid w)}.
\end{equation}
Now, by the definition of \(Y\), \(T\), and \(C\), we have
\begin{align*}
    R_0(t \mid w) &= P_0(Y \geq t \mid W = w) = P_0(T \geq t, C \geq t \mid W = w)
\end{align*}
for all \(t\). By Condition \ref{id: cond1}, we thus have
\begin{align*}
    R_0(t \mid w) &= P_0(T \geq t \mid W = w) P_0(C \geq t \mid W = w) \\
    &= S_0(t^- \mid w) G_{0}(t^- \mid w)
\end{align*}
for each \(t \in (0, t_0]\), recall
\[
G_{0}(t \mid w) := P_0(C > t \mid W = w).
\]
We also have
\begin{align*}
    F_{0,1}(t \mid w) &= P_0(Y \leq t, \Delta = 1 \mid W = w) \\
    &= P_0(\min\{T, C\} \leq t, T \leq C \mid W = w) \\
    &= \int_{u \in (0,t]} \int_{v \geq u} P_0(du, dv \mid W = w),
\end{align*}
where \((u, v) \mapsto P_0(u, v \mid W = w) := P_0(T \leq u, C \leq v \mid W = w)\) is the joint distribution function of \(T\) and \(C\) given \(W = w\). By Condition \ref{id: cond1}, we have
\begin{align*}
    P_0(u, v \mid W = w) &= P_0(T \leq u, C \leq v \mid W = w) \\
    &= P_0(T \leq u \mid W = w) P_0(C \leq v \mid W = w) \\
    &= \left[1 - S_0(u \mid w)\right] \left[1 - G_{0}(v \mid w)\right]
\end{align*}
for each \(u, v \in (0, t_0)\). It follows that
\[
F_{0,1}(t \mid w) = -\int_{u \in (0,t]} G_{0}(u^- \mid w) S_0(du \mid w)
\]
for each \(t \in (0, t_0)\). 
Now, we note that Condition \ref{id: cond2} implies that \(G_{0}(t^- \mid w) > 0\) for \(P_0\)-almost every \(w\) and all \(t \in (0, t_0)\). Therefore, for each \(L(w) \in (0, t]\),
\begin{align}\label{eq:lambda}
\Lambda_{0}(L(w) \mid w) &= -\int_{(0,L(w)]} \frac{S_0(du \mid w)}{S_0(u^- \mid w)} 
= -\int_{(0,L(w)]} \frac{G_{0}(u^- \mid w) S_0(du \mid w)}{G_{0}(u^- \mid w) S_0(u^- \mid w)} \nonumber\\
&= \int_{(0,L(w)]} \frac{F_{0,1}(du \mid w)}{R_0(u \mid w)}.
\end{align}
Combining Eq. \eqref{eq:s} and Eq. \eqref{eq:lambda} completes the proof.   
\end{proof}

\begin{proof}[Proof of Theorem~\ref{thm: eif}]
Note that our target functional can be written in the form of Riemann-Stieltjes product integral \citep{gill1990survey}
\begin{align*}
    \Psi(P_0; L) &= P_{0}(T > L(W))= E_0(E_0(I(T > L(W) \mid W)) = E_0(S_{0}(L(W) \mid W)) \\
    &= E_0 \left( \Prodi_{(0,L(w)]} [1 - \Lambda_0(du \mid w)] \right)
\end{align*}
for $\Lambda_0(t \mid w) = - \int_{(0, t]} \frac{S_0(du \mid w)}{S_0(u^- \mid  w)}$ defined in Eq.~\eqref{eq:lambda_int}. 

Let $Q_0$ denote the marginal distribution of $W$ as implied by $P_0$. 
Let $\{P_\epsilon: \epsilon\}$ be a regular one-dimensional parametric submodel indexed by parameter $\epsilon \in \mathbbm{R}$, through $P_0$ at $\epsilon=0$ in a nonparametric model for the joint distribution of $(W,\Delta,Y)$. We use subscript $\epsilon$ to denote the induced components of $P_\epsilon$ similar to those of $P_0$. For example, we use $Q_\epsilon$ to denote the induced marginal distribution of $W$ under $P_\epsilon$. Let
$\dot{\ell}_{\epsilon}(w) = \frac{\partial}{\partial\epsilon}\log \{dQ_{\epsilon}/dQ_0\}(w)$ denote the score function
for $Q_\epsilon$. 
Thus under appropriate boundedness conditions, 
\begin{align}
    \left. \frac{\partial}{\partial\epsilon}\Psi(P_\epsilon) \right|_{\epsilon=0} &= \left.\frac{\partial}{\partial\epsilon}\int S_{\epsilon}(L(w) \mid w) \, dQ_{\epsilon}(w)\right|_{\epsilon = 0} \nonumber \\
    &= \left.\int \frac{\partial}{\partial\epsilon}S_{\epsilon}(L(w) \mid w)\right|_{\epsilon = 0} \, dQ_0(w) + \int S_0(L(w) \mid w) \,\dot{\ell}_0(w) dQ_0(w). \label{eq:full}
\end{align} 
By definition, the integrand in the first term is
\begin{equation*}
    \left.\frac{\partial}{\partial\epsilon}\Prodi_{(0,L(w)]} \{1 - \Lambda_{\epsilon}(du \mid w)\} \right|_{\epsilon = 0}.
\end{equation*}
By Theorem 8 in \citep{gill1990survey},  
\begin{equation*}
    \left.\frac{\partial}{\partial\epsilon}\Prodi_{(0,L(w)]} \{1 - \Lambda_{\epsilon}(du \mid w)\} \right|_{\epsilon = 0} = -S_0(L(w) \mid w) \left.\int_{(0,L(w)]} \frac{S_0(u^-\mid w)}{S_0(u\mid w)} \frac{\partial}{\partial\epsilon} \Lambda_{\epsilon} (du \mid w)\right|_{\epsilon = 0}.
\end{equation*}
By Eq.~\eqref{eq:lambda},
\begin{align*}
&\frac{\partial}{\partial \epsilon} \Lambda_\epsilon(t \mid w) \Bigg|_{\epsilon = 0} = 
\frac{\partial}{\partial \epsilon} \int_{(0,t]} R_\epsilon(u \mid  w)^{-1} F_{\epsilon, 1}(du \mid  w) \Bigg|_{\epsilon = 0} \\&= 
\int_{(0,t]} R_0(u \mid  w)^{-1} \frac{\partial}{\partial \epsilon} F_{\epsilon, 1}(du \mid  w) \Bigg|_{\epsilon = 0} 
- \int_{(0,t]} \frac{\partial}{\partial \epsilon} R_\epsilon(u \mid  w) \Bigg|_{\epsilon = 0} R_0(u \mid  w)^{-2} F_0(du \mid  w).
\end{align*}
Then
\begin{equation*}
    \frac{\partial \Lambda_\epsilon(du \mid  w)}{\partial \epsilon} \Bigg|_{\epsilon=0} = \frac{\left.\frac{\partial}{\partial \epsilon} F_{\epsilon,1}(du \mid w)\right|_{\epsilon=0}}{R_0(u \mid w)} - \frac{\left.\frac{\partial}{\partial \epsilon} R_\epsilon(u \mid w)\right|_{\epsilon=0} F_0(du \mid w)}{ R_0(u \mid w)^2}.
\end{equation*}
In addition, let $\dot{\ell}_\epsilon(\delta,y \mid w)$ denote the score function for the induced distribution of $(\Delta,Y) \mid W$ under the submodel $P_\epsilon$, that is, the gradient of the Radon-Nikodym derivative (with respect to $\epsilon$) of the conditional distribution of $(\Delta,Y) \mid W=w$ under $P_\epsilon$ relative to that under $P_0$. We have that
\begin{align*}
    \frac{\partial}{\partial \epsilon} F_{\epsilon,1}(u \mid w) \Bigg|_{\epsilon=0} &= 
\frac{\partial}{\partial \epsilon} P_{\epsilon}(Y \leq u, \Delta = 1 \mid  W = w) \Bigg|_{\epsilon=0} \\&= 
\frac{\partial}{\partial \epsilon} \iint I(y \leq u, \delta = 1) P_{\epsilon}(dy, d\delta \mid w) \Bigg|_{\epsilon=0} \\&=
\int \int I(y \leq u, \delta = 1) \dot{\ell}_0(y, \delta \mid w) P_0(dy, d\delta \mid w).
\end{align*}
Thus we have
\begin{align*}
    \frac{\partial}{\partial \epsilon} F_{\epsilon, 1}(du \mid w) \Bigg|_{\epsilon=0} &= \int_\delta I(\delta = 1) \dot{\ell}_0(u, \delta \mid w) P_0(du, d\delta \mid w)\\
    \frac{\partial}{\partial \epsilon} R_\epsilon(u \mid w) \Bigg|_{\epsilon=0} &= \iint I(y \geq u) \dot{\ell}_0(y, \delta |  w) P_0(dy, d\delta |  w),
\end{align*}
where $\int_x$ denotes integration with respect to variable $x$.

Putting these together, the first term in Eq.~\eqref{eq:full}
is 
\begin{align*}
    &\left.\frac{\partial}{\partial\epsilon}\int_w\Prodi_{(0,L(w)]} \{1 - \Lambda_{\epsilon}(du \mid w)\} dQ_0(w)\right|_{\epsilon = 0}\\
    =&\left.\int_w \frac{\partial}{\partial\epsilon}\Prodi_{(0,L(w)]} \{1 - \Lambda_{\epsilon}(du \mid w)\}\right|_{\epsilon = 0} dQ_0(w)\\
    =&\int_w -S_0(L(w) | w) \int_{(0,L(w)]}\frac{S_0(u^-\mid w)}{S_0(u\mid w)R_0(u \mid w)}  \int_\delta I(\delta = 1) \dot{\ell}_0(u, \delta \mid w) P_0(du, d\delta \mid w)dQ_0(w)\\
    -&\int S_0(L(w)| w) \int_{(0,L(w)]}\frac{S_0(u^-\mid w)}{S_0(u\mid w)}\frac{F_0(du \mid w)}{ R_0(u \mid w)^2}  \int_y \int_\delta I(y \geq u) \dot{\ell}_0(y, \delta |  w) P_0(dy, d\delta |  w)dQ_0(w)\\
    =&\int_w\int_{y}\int_{\delta}  -I(y \leq L(w), \delta = 1) \frac{S_0(L(w)\mid w) S_0(y^- \mid w)}{S_0(y \mid w) R_0(y \mid w)} {\dot{\ell}_0(y, \delta \mid w)}  P_0(dy, d\delta \mid w) \, dQ_0(w)\\
    &+\int_w\int_y\int_u\int_\delta  I(u \leq L(w), u \leq y) \cdot \\&\quad\quad\quad \frac{S_0(L(w) \mid w) S_0(u^- \mid w)}{S_0(u \mid w) R_0(u \mid w)^2}{\dot{\ell}_0(y, \delta \mid w)}  P_0(dy, d\delta \mid w)F_0(du\mid w) \, dQ_0(w)\\
    =&\iiiint  -I(y \leq L(w), \delta = 1) \frac{S_0(L(w)\mid w) S_0(y^- \mid w)}{S_0(y \mid w) R_0(y \mid w)} {\dot{\ell}_0(y, \delta \mid w)}  P_0(dy, d\delta \mid w) \, dQ_0(w)
    \end{align*}
    \begin{align*}
    &+\iiint S_0(L(w) \mid w) \int_{(0,L(w)\wedge y]}  \frac{ S_0(u^- \mid w) F_0(du\mid w)}{S_0(u \mid w) R_0(u \mid w)^2}{\dot{\ell}_0(y, \delta \mid w)}  P_0(dy, d\delta \mid w) \, dQ_0(w)\\
    =&E_0 \left[ S_0(L(W)\mid W ) \{ H_0(L(W) \wedge Y, W ) - \frac{I(Y \leq L(W), \Delta = 1)S_0(Y^- | W )}{S_0(Y |W )R_0(Y | W )} \} \dot{\ell_0}(Y, \Delta |W ) \right]
\end{align*}
where $H_0(v, w) := \int_{(0,v]} \frac{S_0(u^- |w) F_0(du | w)}{ S_0(u | w) R_0(u | w)^2}
$.

Now, we observe that:
\[
E_0\left[ I(Y \leq t, \Delta = 1) \frac{S_0(Y^{-} \mid W)}{S_0(Y \mid W) R_0(Y \mid W)} \,\middle|\, W = w \right] = 
\int_{(0,t]} \frac{S_0(y^{-} \mid w)}{S_0(y \mid w) R_0(y \mid w)} F_0(dy \mid w).
\]

Additionally,
\begin{align*}
    E_0\left[ H_0(t \wedge Y, W) \,\middle|\, W = w \right] =&
\int \int_{(0,t]} I(u \leq y) \frac{S_0(u^{-} \mid w)}{S_0(u \mid w) R_0(u \mid w)^2} F_0(du \mid w) P_0(dy \mid w)\\
=&\int_{(0,t]} \frac{S_0(u^{-} \mid w)}{S_0(u \mid w) R_0(u \mid w)} F_0(du \mid w),
\end{align*}
since by definition \( P_0(Y \geq u \mid W = w) = R_0(u \mid w) \).

Therefore:
\[
E_0\left[ H_0(t \wedge Y, W) - I(Y \leq t, \Delta = 1) \frac{S_0(Y^{-} \mid W)}{S_0(Y \mid W) R_0(Y \mid W)} \,\middle|\, W \right] = 0 \quad P_0\text{-almost surely}.
\]
Let $\dot{\ell}_\epsilon(y,\delta,w) := \frac{\partial}{\partial \epsilon} \log \frac{P_\epsilon}{P_0}(y,\delta,w)$ denote the score function for $P_\epsilon$.
By the definition of the score function and that the joint distribution of $(W,Y,\Delta)$ can be written as the product of the marginal distribution of $W$ and the conditional distribution of $(Y,\Delta) \mid W$, we have that
\begin{align*}
    \dot{\ell}_0(y, \delta\mid w) &= \frac{\partial}{\partial\epsilon}\log \frac{P_\epsilon (dy, d\delta\mid w)}{P_0(dy,d\delta \mid w)}(y,\delta)\bigg|_{\epsilon=0} = \frac{\partial}{\partial \epsilon} \left\{\log \frac{d P_\epsilon}{d P_0}(y,\delta,w) - \log \frac{d Q_\epsilon}{d Q_0}(w) \right\} \bigg|_{\epsilon=0} \\
    &= \dot{\ell}_0(y, \delta, w) - \dot{\ell}_0(w)
\end{align*}
and that $E[\dot{\ell}_0(Y,\Delta,W) | W = w] = \frac{\partial}{\partial \epsilon} \log \frac{Q_{\epsilon}}{Q_0}(w)\}\big|_{\epsilon=0}$.
Combing with the Tower rule, 
\begin{align*}
    &\left.\frac{\partial}{\partial\epsilon}\int\Prodi_{(0,L(w)]} \{1 - \Lambda_{\epsilon}(du \mid w)\} dQ_0(w)\right|_{\epsilon = 0} \\
    =&E_0 \left[ S_0(L(W) | W ) \{ H_0(L(W) \wedge Y, W ) - \frac{I(Y \leq L(W), \Delta = 1)S_0(Y^- | W )}{S_0(Y |W )R_0(Y | W )} \} \dot{\ell_0}(Y, \Delta, W ) \right] \\
    -&E_0\left[E_0 \left[ S_0(L(W) | W ) \{ H_0(L(W) \wedge Y, W ) - \frac{I(Y \leq L(W), \Delta = 1)S_0(Y^- | W )}{S_0(Y |W )R_0(Y | W )} \} \dot{\ell}_0(W)| W \right]\right]\\
    =&E_0 \left[ S_0(L(W)  | W ) \{ H_0(L(W)  \wedge Y, W ) - \frac{I(Y \leq L(W) , \Delta = 1)S_0(Y^- | W )}{S_0(Y |W )R_0(Y | W )} \} \dot{\ell_0}(Y, \Delta, W ) \right] \\
    &-E_0 \left[ 0\cdot E_0 \left[  S_0(t | W )\right]\cdot E_0\left[\ell_0(Y,\Delta,W) \mid W\right] \right] \\
    =&E_0 \left[ S_0(L(W)  | W ) \{ H_0(L(W)  \wedge Y, W ) - \frac{I(Y \leq L(W) , \Delta = 1)S_0(Y^- | W )}{S_0(Y |W )R_0(Y | W )} \} \dot{\ell_0}(Y, \Delta, W ) \right].
\end{align*}
Thus the uncentered influence function $\phi_0(y, \delta, w)$ is 
$$S_0(L(w) \mid w) \left[ 1 - \left\{ \frac{I(y \leq L(w) , \delta = 1)S_0(y^- \mid w)}{ S_0(y | w) R_0(y \mid  w) }-\int_{(0,L(w)\wedge y]} \frac{S_0(u^- \mid w) F_0(du \mid w) }{S_0(u\mid w) R_0(u \mid w)^2 }\right\} \right].
$$
Since $\frac{F_0(du \mid w)}{R_0(u\mid w)} = \Lambda_0(du \mid w) \quad \text{and} \quad R_0(u \mid w) = S_0(u^- \mid w)G_0(u \mid w)
$, we have that the uncentered influence function $\phi_0(y, \delta, w)$ equals
\begin{equation*}
    S_0(L(w) \mid w) \left[ 1 - \left\{ \frac{I(y \leq L(w) , \delta = 1)}{S_0(y \mid w)G_0(y \mid w)} - \int_{(0,L(w)\wedge y]} \frac{\Lambda_0(du \mid w)}{ S_0(u \mid w)G_0(u \mid w)} \right\} \right].
\end{equation*}
The influence function $\phi_{0} - \Psi(P_0;L)$ is the efficient influence function under a nonparametric model.
\end{proof}


\subsection{Consistency and asymptotic normality}\label{app:thm}
\begin{lemma}\label{lemma1}
For any conditional survival function \( S \) and corresponding cumulative hazard \( \Lambda \), any conditional censoring function \( G \), 
any LPB $L_{n,\tau}$,
\begin{align}
    &P_0 \phi(S,G;L_{n,\tau}) - \Psi(P_0; L_{n,\tau})= \nonumber \\
    &E_0 \left[ S({L_{n,\tau}(W)} \mid  W) \int_{(0,L_{n,\tau}(w)]} 
\frac{S_0(u^- \mid  W)}{S(u \mid  W)} \left( \frac{ G_0(u \mid  W)}{G(u \mid  W)} - 1 \right) (\Lambda - \Lambda_0)(du \mid  W) \right]. \label{eq:remainder}
\end{align}
\end{lemma}

\begin{proof}
We first write
\[
\phi(S,G;L_{n,\tau})(w, \delta, y) = S({L_{n,\tau}(w)} \mid  w) \left\{ 1 - H_{S,G,n,\tau}(y, \delta, w) \right\},
\]
where we define
\[
H_{S,G,n,\tau}(y, \delta, w) := \frac{I(y \leq {L_{n,\tau}(w)}, \delta = 1) }{S(y \mid  w)G(y \mid  w)} - \int_{(0,L_{n,\tau}(w)\wedge y]} \frac{\Lambda(du \mid  w)} {S(u \mid  w)G(u \mid  w)}.
\]

\noindent We note that
\begin{align*}
    E_0\left[ H_{S,G,n,\tau}(Y, \Delta, W)\mid W = w \right] &= \int_{(0,L_{n,\tau}(w)]} \frac{S_0(y^- \mid  w) G_0(y \mid  w)}{S(y \mid  w) G(y \mid  w)} \Lambda_0(dy \mid  w) \\
    &\quad- \int_{(0,L_{n,\tau}(w)]} \frac{S_0(u^- \mid  w) G_0(u \mid  w)}{S(u \mid  w) G(u \mid  w)} \Lambda(du \mid  w) \\
    &= - \int_{(0,L_{n,\tau}(w)]} \frac{S_0(y^- \mid  w) G_0(y \mid  w)}{S(y \mid  w) G(y \mid  w)} (\Lambda - \Lambda_0)(dy \mid  w).
\end{align*}
Therefore, $P_0 \phi(S,G;L_{n,\tau}) - \Psi(P_0; L_{n,\tau})$ equals
\begin{align*}
&E_0 \left[ S({L_{n,\tau}(W)} |  W) \left\{ 1 - H_{S,G,n,\tau}(Y, \Delta, W) \right\} - S_0({L_{n,\tau}(W)} |  W) \right] \\
=& E_0 \left[ S({L_{n,\tau}(W)} |  W) \int_{(0,L_{n,\tau}(w)]} \frac{S_0(y^- \mid  W) G_0(y \mid  W)}{S(y \mid  W) G(y \mid  W)} (\Lambda - \Lambda_0)(dy \mid W) \right] \\
&\quad + E_0 \left[ S({L_{n,\tau}(W)} |  W) - S_0({L_{n,\tau}(W)} |  W) \right].
\end{align*}

\noindent Now, in view of the Duhamel equation \citep[Theorem~6 of][]{gill1990survey},
we have
\[
S({L_{n,\tau}(w)} |  w) - S_0({L_{n,\tau}(w)} |  w) = -S({L_{n,\tau}(w)} |  w) \int_{(0,L_{n,\tau}(w)]} \frac{S_0(y^- |  w)}{S(y |  w)} (\Lambda - \Lambda_0)(dy |  w),
\]
for each \( w \). Therefore, combining the two terms above yields Eq.~\eqref{eq:remainder}.
\end{proof}

Hereafter, we denote $\phi_{n,\tau}:=\phi(S_n,G_n;L_{n,\tau})$ and $\phi_{\infty,n,\tau}:=\allowbreak\phi(S_\infty,G_\infty;L_{n,\tau})$.
\begin{lemma}\label{lemma2}
    Denote $\Lambda_{\infty}$ and $G_{\infty}$ as the probabilistic limits of respective nuisance estimators $\Lambda_n$ and $G_n$. If Condition~\ref{cond: C6} holds, namely $\Lambda_\infty=\Lambda_0$ or $G_\infty=G_0$, then \( P_0 \phi_{\infty, n, \tau} = \Psi(P_0; L_{n,\tau}) \) and 
\begin{align*}
    \Psi(\hat{P}_n;L_{n,\tau}) - \Psi(P_0; L_{n,\tau}) &= P_n D(P_0, G_{\infty}, S_{\infty}, L_{n,\tau}) + n^{-1/2} \mathbbm{G}_n (\phi_{n,\tau} - \phi_{\infty,n,\tau}) \\&+ [P_0 \phi_{n,\tau} - \Psi(P_0; L_{n,\tau})].
\end{align*}
Here, recall that $\mathbbm{G}_n := n^{1/2} (P_n - P_0)$ denotes the empirical process. 
\end{lemma}

\begin{proof} By Lemma~\ref{lemma1}, \( P_0 \phi_{\infty,n, \tau} - \Psi(P_0; L_{n,\tau}) \) equals
\[
E_0 \left[ S_\infty(L_{n,\tau}(W) |  W ) \int_{(0,L_{n,\tau}(w)]} \frac{S_0(u^{-} |  W )}{ S_\infty(u |  W )} \left\{ \frac{ G_0(u |  W )}{G_\infty(u |  W )} - 1 \right\} ( \Lambda_\infty - \Lambda_0)(du |  W ) \right] .
\]
If \ref{cond: C6} holds, then either \( ( \Lambda_\infty - \Lambda_0)(du |  w) = 0 \), 
or $ \frac{ G_0(u |  W )}{G_\infty(u |  W )} - 1 =0$.
Hence, \( P_0 \phi_{\infty, n,\tau} =\Psi(P_0; L_{n,\tau})\).

To establish the second part of the claim, we observe that
\begin{align*}
&\Psi_{\tau}(\hat{P}_n;L_{n,\tau}) - \Psi(P_0;L_{n,\tau}) \\
&= \frac{1}{|I_{\mathrm{cal}}|}  \sum_{i \in I_{\mathrm{cal}}} \phi_{n,\tau} (O_i) - \Psi(P_0;L_{n,\tau}) \\
&= P_n \phi_{\infty, n,\tau} - \Psi(P_0;L_{n,\tau}) + \frac{1}{|I_{\mathrm{cal}}|}  \sum_{i \in I_{\mathrm{cal}}} \phi_{n,\tau} (O_i) - P_n \phi_{\infty, n,\tau}\\
&= P_n D(P_0, G_{\infty}, S_{\infty}, L_{n,\tau}) + \frac{1}{|I_{\mathrm{cal}}|}  \sum_{i \in I_{\mathrm{cal}}} \left[\phi_{n,\tau} (O_i) - \phi_{\infty, n,\tau} (O_i)\right] \\
&= P_n D(P_0, G_{\infty}, S_{\infty}, L_{n,\tau}) +   P_n (\phi_{n,\tau} - \phi_{\infty, n,\tau}) \\
&= P_n D(P_0, G_{\infty}, S_{\infty}, L_{n,\tau}) +   (P_n - P_0) (\phi_{n,\tau} - \phi_{\infty, n,\tau}) + P_0 (\phi_{n,\tau} - \phi_{\infty, n,\tau}) \\
&= P_n D(P_0, G_{\infty}, S_{\infty}, L_{n,\tau}) + n^{-1/2} \mathbbm{G}_n (\phi_{n,\tau} - \phi_{\infty, n,\tau}) +  \left[P_0 \phi_{n,\tau} - \Psi(P_0; L_{n,\tau})\right].
\end{align*} 
\end{proof}

\noindent Next, we provide a bound on the $L_2(P_0)$ distance between the estimated influence function and its probabilistic limit.

\begin{lemma}\label{lemma3}
If \ref{cond: C2} holds, there exists a universal constant \( C(\eta) \) that may depend on $\eta$ such that, for each \( n \) and \( \tau \),
\begin{align*}
    \left\{ P_0 \left( \phi_{n,\tau} - \phi_{\infty,n,\tau} \right)^2 \right\}^{1/2} &\leq C(\eta) \left( A_{1,n,\tau} + A_{2,n,\tau} \right), \\
    \left\{ P_0 \left\{ \sup_{u \in [0,\tau]} \left| \phi_{n,u} - \phi_{\infty,n,u} \right| \right\}^2 \right\}^{1/2} &\leq C(\eta) \left( A^*_{1,n,\tau} 
+ A_{2,n,\tau} \right),
\end{align*}
where
\begin{align*}
    A_{1,n,\tau}^2 &:= E_0 \left[ \sup_{u \in [0,\tau]} \left| \frac{S_{n}(L_{n, \tau}(W) |  W)}{S_{n}(L_{n, u}(W) |  W)} - \frac{S_{\infty}(L_{n, \tau}(W) |  W)}{S_{\infty}(L_{n, u}(W) |  W)} \right|^2 \right], \\
    A_{2,n,\tau}^2 &:= E_0 \left[ \sup_{u \in [0,\tau]} \left| \frac{1}{G_{n}(L_{n, u}(W)|  W)} - \frac{1}{G_{\infty}(L_{n, u}(W) |  W)} \right|^2 \right], \\
    A^{*2}_{1,n,\tau} &:= E_0 \left[ \sup_{u \in [0,\tau]} \sup_{v \in [0,u]} \left| \frac{S_{n}(L_{n,u}(W) |  W)}{S_{n}(L_{n,v}(W)|  W)} - \frac{S_{\infty}(L_{n, u}(W) |  W)}{S_{\infty}(L_{n, v}(W) |  W)} \right|^2 \right].
\end{align*}
\end{lemma}
\begin{proof}
First, we decompose
\[
\phi_{n,\tau} - \phi_{\infty,n,\tau} = \sum_{j=1}^5 U_{j,n,\tau},
\]
where we define pointwise
{\small\begin{align*}
U_{1,n,\tau}(o) &:= S_{n}(L_{n,\tau}(w) |  w) - S_{\infty}(L_{n,\tau}(w) |  w),\\
U_{2,n,\tau}(o) &:= - \frac{I(y \leq L_{n, \tau}(w), \delta = 1) }{ G_{\infty}(y | w)} \left\{ \frac{S_{n}(L_{n,\tau}(w) | w)}{S_{n}(y | w)} - \frac{S_{\infty}(L_{n,\tau}(w) | w)}{S_{\infty}(y | w)} \right\},\\
U_{3,n,\tau}(o) &:= - I(y \leq L_{n, \tau}(w), \delta = 1) \frac{S_{n}(L_{n,\tau}(w) |  w)}{S_{n}(y |  w)} \left\{ \frac{1}{G_{n}(y |  w)} - \frac{1}{G_{\infty}(y |  w)} \right\},\\
U_{4,n,\tau}(o) &:= \int_{(0,L_{n, \tau}(w) \wedge y]} \left\{ \frac{1}{G_{n}(u |  w)} - \frac{1}{G_{\infty}(u |  w)} \right\} \frac{S_{\infty}(L_{n,\tau}(w) |  w) \Lambda_{\infty}(du |  w)}{S_{\infty}(u |  w)},\\
U_{5,n,\tau}(o) &:= \int_{(0,L_{n, \tau}(w) \wedge y]} \frac{1}{G_n(u |  w)} \left\{ \frac{S_{n}(L_{n,\tau}(w) |  w) \Lambda_{n}(du |  w)}{S_{n}(u |  w)} - \frac{S_{\infty}(L_{n,\tau}(w) |  w) \Lambda_{\infty}(du |  w)}{S_{\infty}(u |  w)} \right\}.
\end{align*}}
By the triangle inequality, we have
\begin{equation*}
    P_0 \left( \phi_{n,\tau} - \phi_{\infty,n,\tau} \right)^2 = \| \phi_{n,\tau} - \phi_{\infty,n,\tau} \|_{2,P_0}^2 \leq \left( \sum_{j=1}^5 \left\| U_{j,n,\tau} \right\|_{2,P_0} \right)^2 = \left\{ \sum_{j=1}^5 (P_0 U_{j,n,\tau}^2)^{1/2} \right\}^2.
\end{equation*}
We bound each term \( P_0 U_{j,n,\tau}^2 \) separately.

\noindent \textbf{Term~1:} Since \( S_{n}(0 |  w) = S_{\infty}(0 |  w) = 1 \) for all \( w\), we have
\begin{align*}
    P_0 U_{1,n,\tau}^2 &= E_0 \left| S_{n}(L_{n,\tau}(W) |  W) - S_{\infty}(L_{n,\tau}(W) |  W) \right|^2 \\&= E_0 \left| \frac{S_{n}(L_{n,\tau}(W) |  W)}{S_{n}(0 |  W)} - \frac{S_{\infty}(L_{n, \tau}(W) |  W)}{S_{\infty}(0 |  W)} \right|^2 \\
    &\leq E_0 \left[ \sup_{u \in [0,\tau]} \left| \frac{S_{n}(L_{n,\tau}(W) |  W)}{S_{n}(L_{n,u}(W)|  W)} - \frac{S_{\infty}(L_{n,\tau}(W) |  W)}{S_{\infty}(L_{n,u}(W) |  W)} \right|^2 \right].
\end{align*}

\noindent \textbf{Term~2:} We have
\begin{align*}
    P_0 U_{2,n,\tau}^2 &= E_0 \left[ \frac{I(Y \leq L_{n, \tau}(W), \Delta = 1) }{G_{\infty}(Y |  W)^2} \left\{ \frac{S_{n}(L_{n,\tau}(W) |  W)}{S_{n}(Y |  W)} - \frac{S_{\infty}(L_{n, \tau}(W) |  W)}{S_{\infty}(Y |  W)} \right\}^2 \right] \\
    &\leq \eta^2 E_0 \left[ \sup_{u \in [0,\tau]} \left| \frac{S_{n}(L_{n,\tau}(W) |  W)}{S_{n}(L_{n,u}(W) |  W)} - \frac{S_{\infty}(L_{n,\tau}(W) |  W)}{S_{\infty}(L_{n,u}(W) |  W)} \right|^2 \right].
\end{align*}

\noindent \textbf{Term~3:} Similarly, since $S_n$ is a survival function, which is non-increasing,
\begin{align*}
    P_0 U_{3,n,\tau}^2 &= E_0 \left[ I(Y \leq L_{n, \tau}(W), \Delta = 1) \frac{S_{n}(L_{n,\tau}(W) | W)^2}{S_{n}(Y | W)^2} \left\{ \frac{1}{G_{n}(Y | W)} - \frac{1}{G_{\infty}(Y | W)} \right\}^2 \right] \\
    &\leq E_0 \left[ \sup_{u \in [0,\tau]} \left| \frac{1}{G_{n}(L_{n, u}(W) | W)} - \frac{1}{G_{\infty}(L_{n,u}(W) | W)} \right|^2 \right].
\end{align*}

\noindent \textbf{Term~4:}
{\small\begin{align*}
P_0 U_{4,n,\tau}^2 &= E_0 \left[ \int_{(0,L_{n, \tau}(W) \wedge Y]} \left\{ \frac{1}{G_{n}(u | W)} - \frac{1}{G_{\infty}(u | W)} \right\} \frac{S_{\infty}(L_{n,\tau}(W) | W) \Lambda_{\infty}(du | W)}{S_{\infty}(u | W)} \right]^2 \\
&\leq E_0 \left[ \sup_{u \in [0,L_{n,\tau}(W)]} \left| \frac{1}{G_{n}(u | W)} - \frac{1}{G_{\infty}(u | W)} \right| \int_{(0,L_{n,\tau}(W) \wedge Y]} \frac{S_{\infty}(L_{n,\tau}(W) | W) \Lambda_{\infty}(du | W)}{S_{\infty}(u | W)} \right]^2\\
&= E_0 \left[ \sup_{u \in [0,\tau]} \left| \frac{1}{G_{n}(L_{n,u}(W) | W)} - \frac{1}{G_{\infty}(L_{n,u}(W) | W)} \right| |1 - S_{\infty}(L_{n,\tau}(W) \wedge Y | W)| \right]^2\\
&\leq E_0 \left[ \sup_{u \in [0,\tau]} \left| \frac{1}{G_{n}(L_{n,u}(W) | W)} - \frac{1}{G_{\infty}(L_{n,u}(W) | W)} \right|^2 \right].
\end{align*}}

\noindent \textbf{Term~5:} We define
\[
B_{n,\tau}(u | w) := \frac{S_{n}(L_{n,\tau}(w) | w)}{S_{n}(u | w)} \quad \text{and} \quad B_{\infty,n,\tau}(u | w) := \frac{S_{\infty}(L_{n,\tau}(w) | w)}{S_{\infty}(u | w)},
\]
and we note that
\[
B_{n,\tau}(du | w) = \frac{S_{n}(L_{n,\tau}(w) | w) \Lambda_{n}(du | w)}{S_{n}(u | w)} \quad \text{and} \quad B_{\infty,n,\tau}(du | w) = \frac{S_{\infty}(L_{n,\tau}(w) | w) \Lambda_{\infty}(du | w)}{S_{\infty}(u | w)}
\]
by the backwards equation \citep[Theorem 5 of][]{gill1990survey}. Thus,
\[
P_0 U_{5,n,\tau}^2 = E_0 \left[ \int_{(0,L_{n,\tau}(W) \wedge Y]} \frac{1}{G_{n}(u | W)} \left\{ B_{n,\tau}(du | W) - B_{\infty,n,\tau}(du | W) \right\} \right]^2.
\]
Using integration by parts, this can be re-expressed as
{\small\begin{align*}
&E_0 \Bigg[ \frac{1}{G_{n}(L_{n,\tau}(W) \wedge Y | W)} \left\{ B_{n,\tau}(L_{n,\tau}(W) \wedge Y | W) - B_{\infty,n,\tau}(L_{n,\tau}(W) \wedge Y | W) \right\} \\&-\{B_{n,\tau}(0 | W) - B_{\infty,n,\tau}(0 | W) \} - \int_{(0,L_{n,\tau}(W) \wedge Y]} \frac{B_{n,\tau}(u | W) - B_{\infty,n,\tau}(u | W) }{G_{n}(u | W)^2}  G_{n}(du | W) \Bigg]^2\\
&\leq \eta^2 E_0 \left[ \left| \frac{S_{n}(L_{n,\tau}(W) | W)}{S_{n}(L_{n,\tau}(W) \wedge Y | W)} - \frac{S_{\infty}(L_{n,\tau}(W) | W)}{S_{\infty}(L_{n,\tau}(W) \wedge Y | W)} \right| + |S_{n}(L_{n,\tau}(W) | W) - S_{\infty}(L_{n,\tau}(W) | W)| \right.\\
&\quad\quad+ \sup_{u \in [0,\tau]} \left. \left| \frac{S_{n}(L_{n,\tau}(W) | W)}{S_{n}(L_{n,u}(W)  | W)} - \frac{S_{\infty}(L_{n,\tau}(W) | W)}{S_{\infty}(L_{n,u}(W) | W)} \right| \frac{1}{G_{n}(L_{n,u}(W) | W)} \right]^2\\
&\leq 3 \eta^4 E_0 \left[ \sup_{u \in [0,\tau]} \left| \frac{S_{n}(L_{n,\tau}(W) | W)}{S_{n}(L_{n,u}(W) | W)} - \frac{S_{\infty}(L_{n,\tau}(W) | W)}{S_{\infty}(L_{n,u}(W) | W)} \right|^2 \right].
\end{align*}}
\end{proof}

Next, we bound the convergence rate of the empirical process term $\mathbbm{G}_n (\phi_{n,\tau} - \phi_{\infty,n,\tau})$.

\begin{lemma}\label{lemma4}
If \ref{cond: C1}--\ref{cond: C2} hold, then
\begin{equation}
    n^{-1/2} \mathbbm{G}_n (\phi_{n,\tau} - \phi_{\infty,n,\tau}) = o_p(n^{-1/2}) \label{eq: empirical process 1}
\end{equation}
and
\begin{equation}
    n^{-1/2}\sup_{u \in [0,\tau]} \left|  \mathbbm{G}_n (\phi_{n,u} - \phi_{\infty,n,u}) \right|=O_p(n^{-1/2}). \label{eq: empirical process 3}
\end{equation}
If \ref{cond: C3} also holds, then
\begin{equation}
    n^{-1/2} \sup_{u \in [0,\tau]} \left| \mathbbm{G}_n (\phi_{n,u} - \phi_{\infty,n,u}) \right| = o_p(n^{-1/2}). \label{eq: empirical process 2}
\end{equation}
\end{lemma}

\begin{proof}
   For Eq.~\eqref{eq: empirical process 1}, it suffices to show that
\[\left| \mathbbm{G}_n \left( \phi_{n,\tau} - \phi_{\infty,n,\tau} \right) \right| = o_p(1).
\]
As \( \mathbbm{G}_n \) does not involve the training data and is hence independent of \( \phi_{n,\tau} - \phi_{\infty,n,\tau} \), 
we condition on the training data first.
By the Central Limit Theorem (CLT), $\mathbbm{G}_n(\phi_{n,\tau} - \phi_{\infty,n,\tau})$ converges weakly to $\mathcal{N}(0, P_0(\phi_{n,\tau} - \phi_{\infty,n,\tau})^2)$.
By Lemma~\ref{lemma3}, Conditions~\ref{cond: C1} and \ref{cond: C3}, and conditional convergence implying unconditional convergence \citep[Lemma~6.1]{Chernozhukov2018},
unconditional on the training data, we have that $E_0 \left| \mathbbm{G}_n \left( \phi_{n,\tau} - \phi_{\infty,n,\tau} \right) \right| = o(1)$ and thus, by Markov's inequality,
$$\left| \mathbbm{G}_n \left( \phi_{n,\tau} - \phi_{\infty,n,\tau} \right) \right| = o_p(1).$$

Next, we prove the uniform bound in Eq.~\eqref{eq: empirical process 2}. For a fixed $u$, $L_{n,u}$ only depends on the training data $\mathcal{D}_{\mathrm{train}}$. 
Conditioning on the training data, Lemma~5 in \citet{westling2023inference} implies that $\{\phi_{n,u} - \phi_{\infty,n,u}: u \in [0, \tau]\}$ is a $P_0$-Donsker class, because bounded uniform entropy integral implies a Donsker class. Thus,
by the fourth displayed equation on page~128 of Theorem~2.5.2 in \citet{vaart1996weak},
there exists $\bar{C} > 0$, such that
\begin{align*}
E_0 &\left\{ \sup_{u \in [0,\tau]} \left| \mathbbm{G}_n \left( \phi_{n,u} - \phi_{\infty,n,u} \right) \right| \right\}\\
&\leq \bar{C}\left\{ E_0  \sup_{u \in [0,\tau]} \left[ \phi_{n,u}(O) - \phi_{\infty,n,u}(O) \right]^2 \right\}^{1/2}
\leq \bar{C} C(\eta) \left( A_{1,n,u}^* + A_{2,n,u} \right),
\end{align*}
where the second inequality follows by Lemma~\ref{lemma3}.

Again by Lemma~6.1 in \citet{Chernozhukov2018},
\begin{equation*}
    E \left[ \sup_{u \in [0,\tau]} \left\lvert \mathbbm{G}_n \left( \phi_{n,u} - \phi_{\infty,n,u} \right) \right\rvert \right]
\leq \bar{C} C(\eta) \left( 
 A^*_{1,n,u}  
+ A_{2,n,u} \right). 
\end{equation*}
By \ref{cond: C1}, \ref{cond: C2}, and \ref{cond: C3}, this bound converges to zero, 
and thus 
$$n^{-1/2} \sup_{u \in [0, \tau]} \left| \mathbbm{G}_n (\phi_{n,u} - \phi_{\infty,n,u}) \right| = o_p(n^{-1/2}).$$

We finally prove the weaker uniform bound in Eq.~\eqref{eq: empirical process 3}. Since $u \mapsto S_n(u \mid w)$ and $u \mapsto S_\infty(u \mid w)$ are positive non-increasing, it holds that
$$\frac{S_{n}(L_{n,u}(w) |  w)}{S_{n}(L_{n,v}(w)|  w)}, \frac{S_{\infty}(L_{n, u}(w) |  w)}{S_{\infty}(L_{n, v}(w) |  w)} \in (0,1]$$
whenever $0 \leq v \leq u$. Hence, the terms $A_{1,n,\tau}$ and $A_{1,n,\tau}^*$ in Lemma~\ref{lemma3} are both bounded above by 1.
Thus, by Lemma~\ref{lemma3} and a similar argument as above,
$$n^{-1/2}\sup_{u \in [0,\tau]} \left|  \mathbbm{G}_n (\phi_{n,u} - \phi_{\infty,n,u}) \right|=O_p(n^{-1/2}).$$

\end{proof}

\begin{proof}[Proof of Theorem~\ref{thm2}]
Lemma~\ref{lemma2} implies that $|\Psi(\hat{P}_n;L_{n,\tau}) - \Psi(P_0; L_{n,\tau})|$ is bounded above by 
\[\left| P_n D(P_0, G_{\infty}, S_{\infty}, L_{n,\tau}) \right| + \left|   n^{-1/2}\mathbbm{G}_n (\phi_{n,\tau} - \phi_{\infty, n,\tau}) \right| + \left| P_0 (\phi_{n,\tau} - \phi_{\infty, n,\tau}) \right|.
\]

By Lemma~\ref{lemma2}, $P_0 \phi_{\infty, n,\tau} = \Psi(P_0; L_{n,\tau})$, thus the first term is an empirical mean of a mean zero function, which by the weak law of large numbers is \( o_p(1) \).
By Lemma \ref{lemma4},
the the second term is also $O_p(n^{-1/2})$.
The third term is bounded by $[ P_0 (\phi_{n,\tau} - \phi_{\infty,n,\tau})^2 ]^{1/2}$, which is $o_p(1)$ by Lemma~\ref{lemma3} and  Condition~\ref{cond: C1}--\ref{cond: C2}.
Thus, \( |\Psi (\hat{P}_n;L_{n,\tau}) - \Psi(P_0;L_{n,\tau})| = o_p(1) \).

For uniform consistency, Lemma~\ref{lemma2} and the triangle inequality yield
\begin{align*}
    \sup_{u \in [0,\tau]} |\Psi (\hat{P}_n;L_{u,\tau})- \Psi(P_0;L_{n,u})| \leq& \sup_{u \in [0,\tau]} \left| P_n D(P_0, G_{\infty}, S_{\infty},L_{n,u}) \right|\\ + \sup_{u \in [0,\tau]} &\left| n^{-1/2} \mathbbm{G}_n (\phi_{n,u} - \phi_{\infty,n,u}) \right| 
+ \sup_{u \in [0,\tau]} \left|  P_0 (\phi_{n,u} - \phi_{\infty,n,u}) \right|.
\end{align*}
By Condition~\ref{cond: C1}, $S_\infty(t_0 \mid w) \geq \zeta$ for $P_0$-almost every $w$.
The first term on the right-hand side of the inequality above 
is $O_p(n^{-1/2})$. This is because, as we will show next, $D(P_0, G_{\infty}, S_{\infty},L_{n,u})$ falls in a $P_0$-Donsker class conditioning on the training data. 
Recall that $D(P_0, G_{\infty}, S_{\infty},L_{n,u}) = \phi(S_{\infty},G_{\infty};L_{n,u}) - \Psi(P_0;L_{n,u})$ in Eq.~\eqref{eq: d_tau}. We start from the first uncentered influence function term, \begin{align*}
    \phi&(S_{\infty},G_{\infty};L_{n,u}) : (w,\delta,y) \mapsto \\ &S_{\infty}(L_{n,u}(w) \mid w) \left[ 1 - \left\{ \frac{I(y \leq L_{n,u}(w) , \delta = 1)}{S_{\infty}(y \mid w)G_{\infty}(y \mid w)} - \int_{(0,L_{n,u}(w)\wedge y]} \frac{\Lambda_{\infty}(dv \mid w)}{ S_{\infty}(v \mid w)G_{\infty}(v \mid w)} \right\} \right].
\end{align*}
First, the class of weighted indication functions $\{(w,\delta,y) \mapsto \frac{I(y \leq L_{n,u}(w) , \delta = 1)}{S_{\infty}(y \mid w)G_{\infty}(y \mid w)}: u \in [0, \tau]\}$ is VC-subgraph \citep[e.g., Problem~20 on Page~153 of][]{vaart1996weak} and therefore BUEI.
Second, consider the class $\{(w,\delta,y) \mapsto \int_{(0,L_{n,u}(w)\wedge y]} \frac{\Lambda_{\infty}(dv \mid w)}{ S_{\infty}(v \mid w)G_{\infty}(v \mid w)}: u \in [0, \tau]\} $ that is a subset of the function class $\{(w,\delta,y) \mapsto \int_{(0,t\wedge y]} \frac{\Lambda_{\infty}(du \mid w)}{ S_{\infty}(u \mid w)G_{\infty}(u \mid w)}: t \in [0, t_0)\}$, which is a subset of a VC-major by Lemma~2.6.19 in \citet{vaart1996weak} and thus also BUEI.
Third, $\{(w,\delta,y) \mapsto S_{\infty}(L_{n,u}(w) \mid w): u\in[0,\tau]\}$ is uniformly bounded and BUEI as well. By applying properties (iv)–(v) of Lemma 9.17 in~\citet{kosorok2008introduction}, the products and sums of uniformly bounded BUEI classes remain BUEI. Consequently, $\{\phi(S_{\infty},G_{\infty};L_{n,u}): u \in [0,\tau]\}$ is a BUEI class.
Furthermore, $\{(w,\delta,y) \mapsto \Psi(P_0;L_{n,u}):u\in[0,\tau]\}$ is VC-subgraph by Lemma~2.6.15 in~\citet{vaart1996weak} and thus also BUEI. Applying (iv) in Lemma 9.17 in~\citet{kosorok2008introduction} again yields that $\{D(P_0, G_{\infty}, S_{\infty},L_{n,u}) : u\in[0, \tau)\}$ is a uniformly bounded BUEI class, and hence is $P_0$-Donsker. 

By Lemma~\ref{lemma4}, the second term is $O_p(n^{-1/2})$. 
By Lemma~\ref{lemma3}, the third term is $o_p(1)$. We thus find that
\[
\sup_{u \in [0,\tau]} |\Psi (\hat{P}_n;L_{n,u}) - \Psi (P_0;L_{n,u})| = o_p(1).
\]  
\end{proof} 

\begin{proof}[Proof of Theorem~\ref{thm3}]
We first consider a fixed tuning parameter $\tau$.
Since Condition~\ref{cond: C6}
is implied by Condition~\ref{cond: C5},
by Lemma~\ref{lemma2},
\begin{align}
    \Psi&(\hat{P}_n;L_{n,\tau}) - \Psi(P_0;L_{n,\tau}) \nonumber \\&= P_n D(P_0, G_0, S_0;L_{n,\tau}) + n^{-1/2} \mathbbm{G}_n (\phi_{n,\tau} - \phi_{0,n,\tau}) + [P_0 \phi_{n,\tau} - \Psi(P_0;L_{n,\tau})].\label{eq:decomp}
\end{align}
By Condition~\ref{cond: C1}--\ref{cond: C3} and \ref{cond: C5}, 
the second term on the right-hand side of Eq.~\eqref{eq:decomp} is $o_p(n^{-1/2})$ by Lemma~\ref{lemma4}. By Lemma~\ref{lemma1},
\begin{align*}
    &P_0 \phi_{n,\tau} - \Psi(P_0;L_{n,\tau}) \\
    &= E_0 \left[ S_n({L_{n,\tau}(W)} \mid  W) \int_{(0,L_{n,\tau}(W)]} 
\frac{S_0(u^- \mid  W)}{S_n(u \mid  W)} \left( \frac{ G_0(u \mid  W)}{G_n(u \mid  W)} - 1 \right) (\Lambda_n - \Lambda_0)(du \mid  W) \right].
\end{align*}
By the Duhamel equation~\citep[Theorem 6]{gill1990survey}, we have that
\[
\frac{S_0(u^- \mid W)}{S_{n}(u \mid W)} (\Lambda_n - \Lambda_0)(du \mid  W) = \left( \frac{S_0}{S_{n}} - 1 \right) (du \mid W),
\]
and so the above equals
\[
E_0 \left[ S_n({L_{n,\tau}(W)} \mid  W) \int_{(0,L_{n,\tau}(W)]} 
\left( \frac{ G_0(u \mid  W)}{G_n(u \mid  W)} - 1 \right) \left( \frac{S_0}{S_{n}} - 1 \right) (du \mid W) \right].
\]
Therefore, 
\[\left|P_0 \phi_{n,\tau} - \Psi(P_0;L_{n,\tau}) \right| = r_{n,\tau}.
\]
which is $o_p(n^{-1/2})$ by Condition~\ref{cond: C4}. 
This establishes that
\[
\Psi(\hat{P}_n;L_{n,\tau}) = \Psi(P_0;L_{n,\tau}) + P_n D(P_0, G_0, S_0,L_{n,\tau}) + o_p(n^{-1/2}).
\]

The above argument applies to $\tau \in [0,1]$ uniformly. In particular, the $o_p(n^{-1/2})$ term is uniform over $\tau \in [0,1]$, so Eq.~\ref{eq:ovvar} holds.

\end{proof}

\section{Additional Simulation Results}\label{app:res}

In this section, we present additional experimental results to illustrate the robustness and efficiency of our method.

\subsection{Marginal, Gaussian process for $C \mid W$}\label{app:res1}
Our empirical coverage and average LPB results with the conditional censoring function estimated with a Gaussian Process (GP) model are shown in Figure~\ref{fig:marg3}. It is worth noting that the GP model is misspecified: the estimated conditional censoring function 
varies little between covariates.
Despite this, our method maintains the most stable coverage at the target \(1-\alpha\) level across all settings. In contrast, all of DFT-fixed, DFT-adaptive-T and DFT-adaptive-CT methods perform well in homoscedastic settings (1 and 2) but degrade in Settings 3 and 4.
This is due to the double robustness of TCsurv against poor estimation of nuisance functions, whereas the other methods are more sensitive to nuisance estimation errors. 

For average LPB, TCsurv is comparable to other methods and excels in the most complex univariate setting (4), where it is one of the only two methods maintaining the target coverage. TCsurv has a higher mean LPB, reflecting its greater efficiency.

\begin{figure}[h]
    \centering
    \begin{subfigure}{\linewidth}
        \centering
        \includegraphics[width=0.9\linewidth]{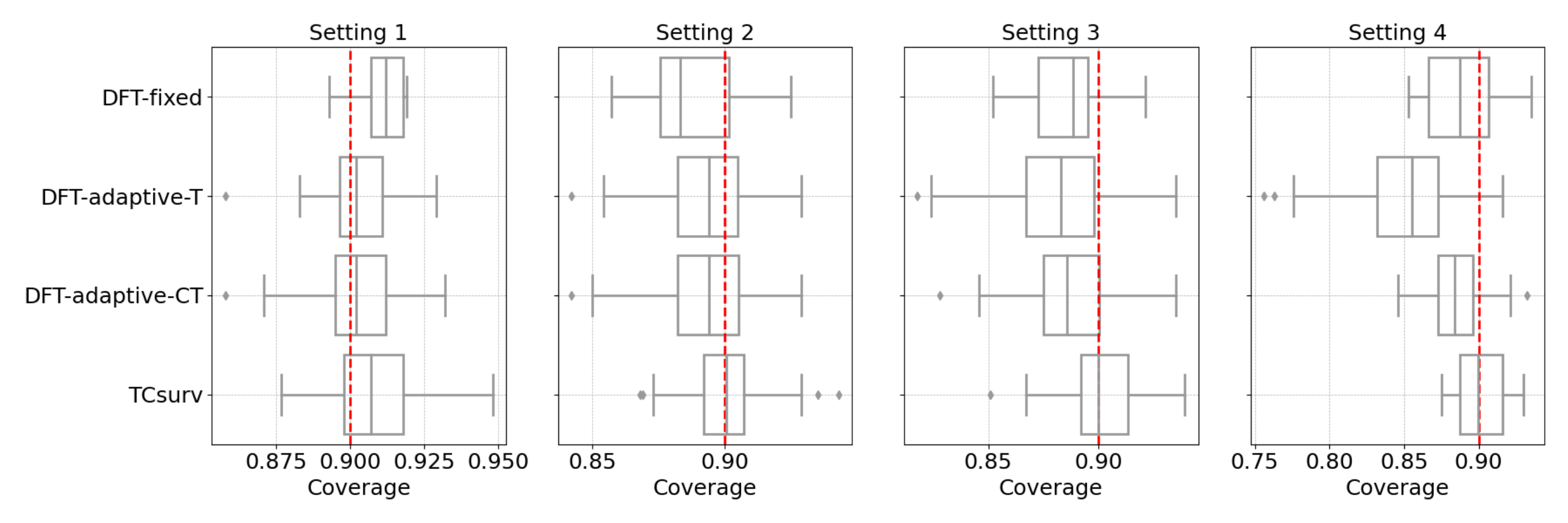}
    \end{subfigure}
    
    \vspace{1em} 

    \begin{subfigure}{\linewidth}
        \centering
        \includegraphics[width=0.9\linewidth]{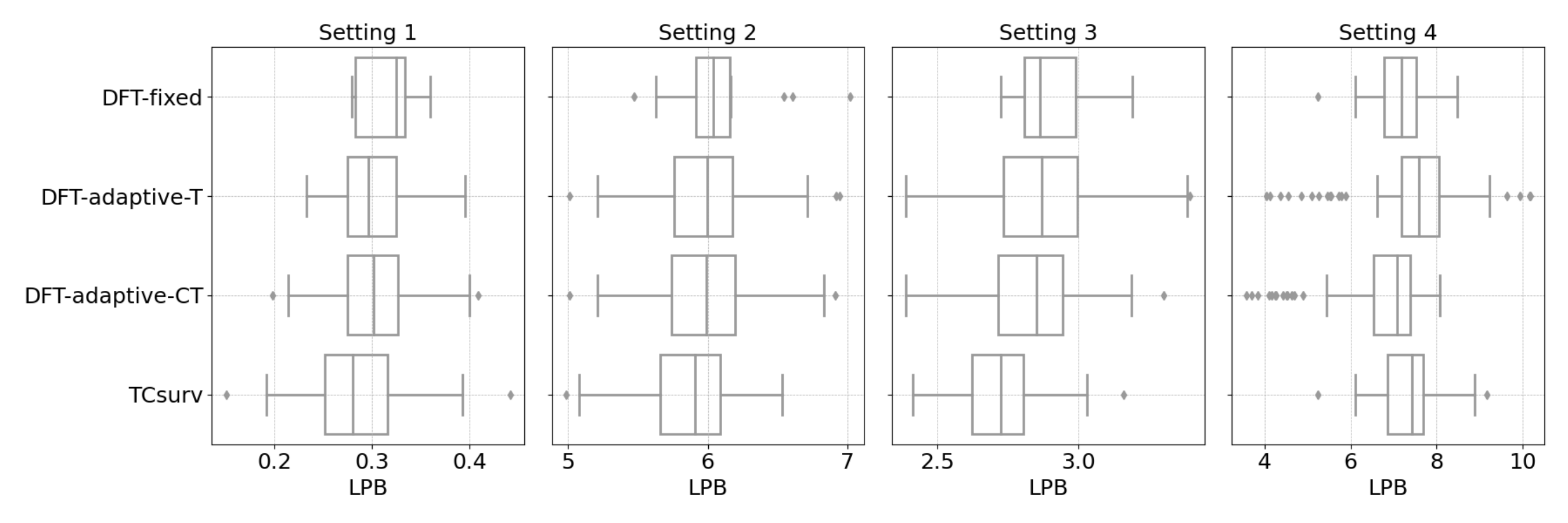}
    \end{subfigure}
    
    \caption{Empirical coverage (top) and average LPBs (bottom) of all the candidate methods under settings 1–4, where $W$ is univariate. The boxplot shows results from 100 independent draws of datasets. The dashed red line corresponds to the target coverage level $1 - \alpha = 90\%$. 
}
    \label{fig:marg3}
\end{figure}

\subsection{Marginal, no access to $C$}\label{app:res2}
We provide the marginal results for the original version of our method (TCsurv-NC), where the true censored time $C$ is unavailable when the event is observed, that is, $T \leq C$. As in previous experiments, $S_n$ is estimated using the \texttt{survSuperlearner} R package, while $G_n$ is modeled with the Cox proportional hazards model. The results show that, TCsurv-NC achieves the target marginal coverage rate and nearly the same interval size as TCsurv, even without full access to $C$. 

\begin{figure}[h]
    \centering
    \begin{subfigure}{\linewidth}
        \centering
        \includegraphics[width=0.9\linewidth]{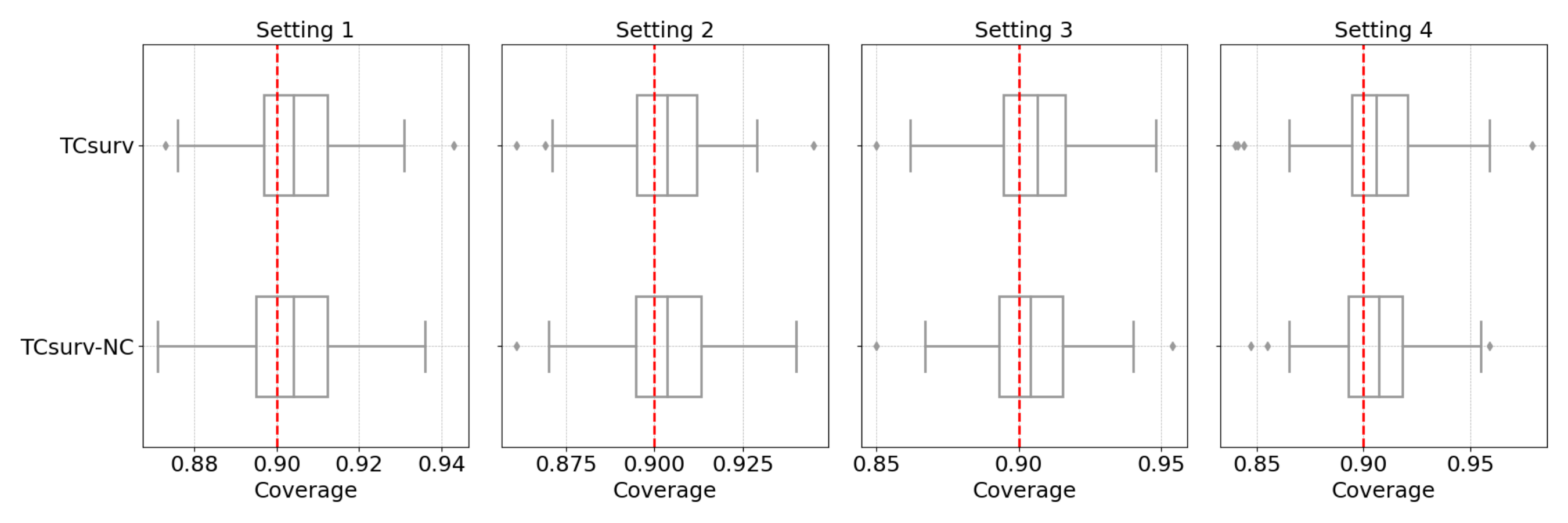}
    \end{subfigure}
    

    \begin{subfigure}{\linewidth}
        \centering
        \includegraphics[width=0.9\linewidth]{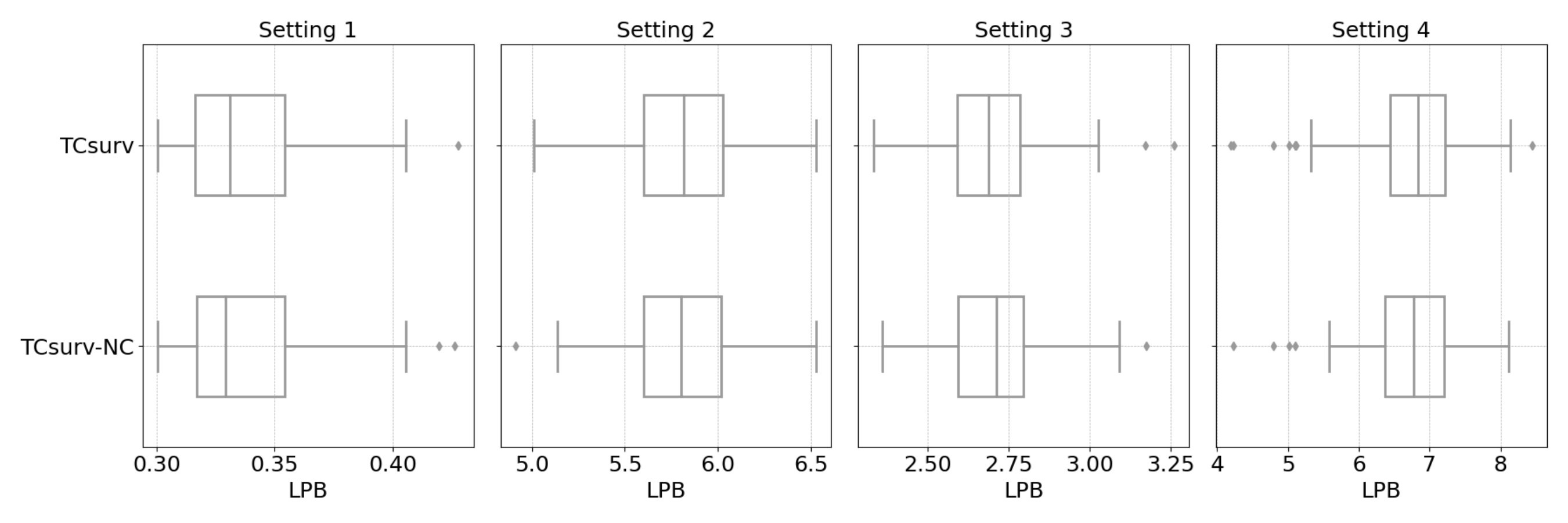}
    \end{subfigure}
    
    \caption{Empirical coverage (top) and average LPBs (bottom) of TCsurv and TCsurv-NC under settings 1–4, where $W$ is univariate. The boxplot shows results from 100 independent draws of datasets. The dashed red line corresponds to the target coverage level $1 - \alpha = 90\%$.
    }
    \label{fig:marg7}
\end{figure}
\begin{figure}[ht!]
    \centering
    \begin{subfigure}[b]{0.45\textwidth}
        \centering
        \includegraphics[width=\textwidth]{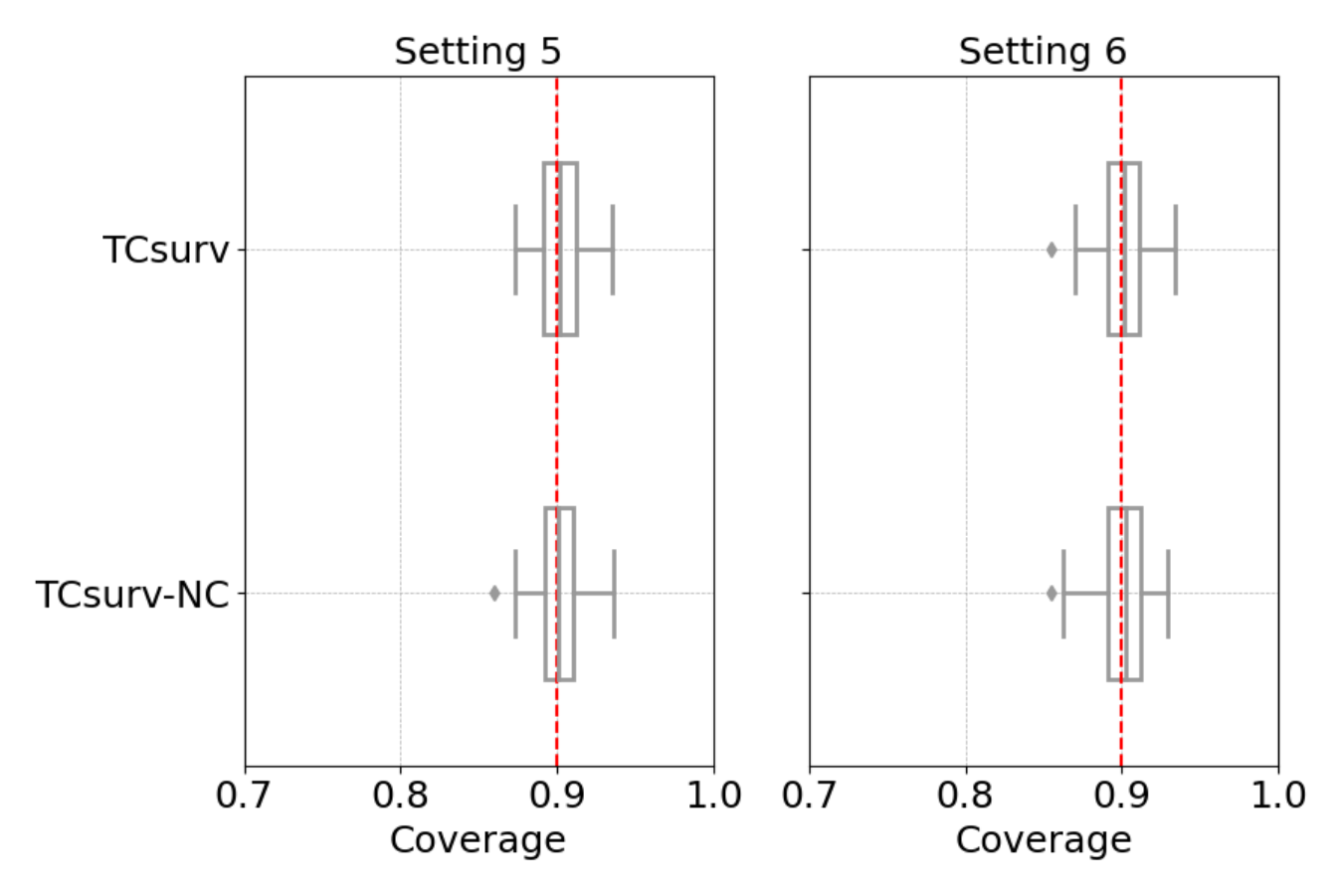}
    \end{subfigure}
    \hfill
    \begin{subfigure}[b]{0.45\textwidth}
        \centering
        \includegraphics[width=\textwidth]{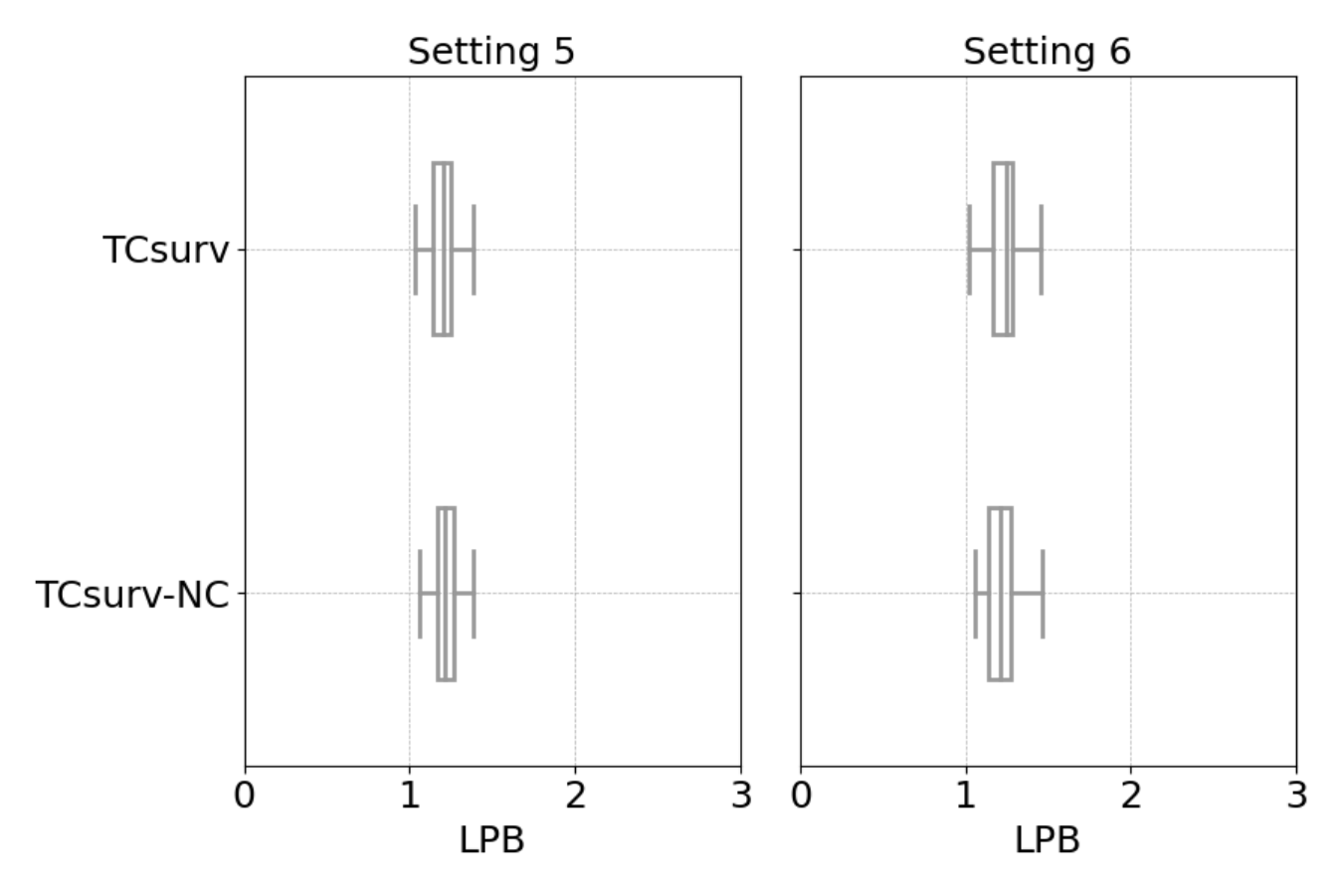}
    \end{subfigure}
    \caption{Comparison of TCsurv and TCsurv-NC on the empirical coverage (left) and average LPBs (right) in the multivariate experimental settings. 
    }
    \label{fig:marg8}
\end{figure}

\end{document}